%% file: main.tex
\documentclass[10pt,aps,pra,twocolumn,showpacs,superscriptaddress,nobalancelastpage,notitlepage,nofootinbib,longbibliography]{revtex4-2}

\input{header}

\input{commands.tex}

\sethlcolor{yellow!50}

\usepackage{xspace}
\usepackage{xparse}
\NewDocumentCommand{\cbox}{s O{1ex}}{%
  \setlength{\fboxsep}{-\fboxrule}%
  \IfBooleanTF{#1}
    {\frame{\rule{0.5\dimexpr#2}{0.5\dimexpr#2}\rule[0.5\dimexpr#2]{0.5\dimexpr#2}{0.5\dimexpr#2}}}
    {\fbox{\rule[0.5\dimexpr#2]{0.5\dimexpr#2}{0.5\dimexpr#2}\rule{0.5\dimexpr#2}{0.5\dimexpr#2}}}
  \xspace
}

\begin{document}

\title{Computing time-dependent reduced models for classical and quantum dynamics}

\author{Tommaso Grigoletto}
\affiliation{Department of Information Engineering, University of Padova, Italy} 

\date{\today}

\begin{abstract}
    This paper introduces a novel method for approximating the dynamics of a large autonomous system projected onto a fixed subspace. The core contribution is a novel recursive algorithm to construct an effective time-dependent generator that is polynomial in the time variable, ensuring accuracy for short time scales. The derivation is based on the Taylor expansion of the exponential map and a new result for computing the time-ordered exponential of polynomial generators.
    This work is motivated by the challenge of deriving time-convolutionless master equations in quantum physics and the proposed method offers an alternative to typical derivations based on expansions in the coupling strength. 
    The resulting approximation is accurate for small times, does not require a weak-coupling assumption, performs better than a truncation of the exponential map at low orders, and crucially, guarantees a completely positive and trace-preserving map at the lowest orders. The proposed method is validated against several prototypical models: a dephasing spin-boson model, a central spin model, and an Ising spin chain.
\end{abstract}

\maketitle

\section{Introduction}
Whenever a dynamical model is too large to be simulated efficiently, one should resort to model reduction techniques to find a smaller model that is simulatable and that retains fundamental characteristics of the original system. A dominant approach in control systems theory -- including, for example, minimal realization techniques \cite{rosenbrock1970state}, balanced truncation \cite{antoulas,gugercin2004survey}, and moment matching \cite{astolfi_moment_matching,padoanmomentmatching} -- consists of finding a low-dimensional subspace where to project the dynamics, obtaining a linear time-invariant (LTI) model. This work considers a different approach: the space onto which the dynamics is projected is assumed to be {\em given a priori and fixed}, and, instead, the dynamics is approximated by seeking a linear time-varying (LTV) model. 

This approach is inspired by quantum physics. 
To obtain realistic models of open quantum systems, one shall consider the effects of the interaction between the system and its environment \cite{breuer2002theory}. Unfortunately, simulating the system-environment couple is impractical (if not impossible), as the scaling of the Hilbert space makes such a simulation prohibitive even for small environments. This implies that, in order to study the effect of the coupling with the environment, one needs to find a reduced model that (hopefully) captures the interesting features of the interaction.

A common method to reduce the joint system-environment model is known in the physics literature with the name of time-convolutionless master equations (TCL ME) \cite{shibata_generalized_1977,chaturvedi_time-convolutionless_1979}. This method aims at obtaining an LTV model of the same size of the system that replicates the effects of the interaction with the environment. Although a formal exact solution to this problem exists, computing this exact LTV generator is generally considered as difficult as solving the original model. The main approach found in the physics literature approximates this formal solution by leveraging an expansion in the strength of the coupling between the system and the environment, which is typically assumed to be small. For more details, see  \cite{breuer2002theory} and references therein. 

In this work an alternative approach is proposed: instead of an expansion in the coupling strength, an expansion in time is considered. In particular, since the formal solution provides an analytic LTV generator in the time variable $t$, the exact incomputable generator is approximated with its truncation to a finite order $N$.  Under these assumptions a recursive formula for the coefficients of the polynomial generator is found, and the approximation error of the resulting reduced model scales as $O(t^{N+1})$. 

Note that while time-dependent dynamical systems are typically more difficult to study than time-independent ones, this approach finds an important application in the simulation of very large systems, as the Runge-Kutta methods allow one to simulate time-dependent models \cite{ascher1997implicit}. In this regard, under the assumption of a time-dependent generator that is {\em polynomial} in $t$, this paper proposes a novel recursive formula to compute the propagator of this LTV model.

Another benefit of the proposed reduction method consists in the fact that, when computing the reduced LTV generator for unitarily evolving quantum systems, the reduced polynomial generator at first and second order is provably in Lindblad form for all times $t$ thus ensuring that the reduced model at the lowest orders generates a completely positive (CP) and trace preserving (TP) semigroup.   

It is important to point out that the method proposed in this work is not intended to substitute the perturbative approach for approximating TCL MEs \cite{shibata_generalized_1977,chaturvedi_time-convolutionless_1979} but rather to provide an alternative that can be of use whenever the main assumptions of the standard method do not hold, e.g. in the strong-coupling regime or whenever the time-ordered integrals are difficult to assess. 

\textbf{Structure and paper contributions:} For the reader's convenience, the structure of this work and its main contributions are summarized next:  
\begin{itemize}
    \item In Sec. \ref{sec:problem_setting} the problem treated in this work is defined rigorously. The known formal solution of the problem is derived in detail using the Nakajima-Swanzig equations \cite{nakajima,zwanzig},  and the key difficulties associated with computing this LTV model are explained.  
    \item In Sec. \ref{sec:time_ordered} the main contribution of the paper is presented: Theorem \ref{thm:approx} derives the polynomial approximation of the formal LTV model and the recursive equation that determines its coefficients. This result is based on Theorem \ref{thm:time_ordered_expansion} where a novel method to compute the time-ordered exponential is derived. The reduction procedure is also validated on a linear model and compared against the truncation of the exponential series. 
    \item Sec. \ref{sec:cptp_second_order} focuses on the application of the proposed reduction procedure to quantum dynamical models. The proposed method is connected and compared with the typical physics approach for the approximation of TCL ME which is based on the expansion of the coupling strength. The connections between the recursive formula obtained here and other similar recursive equations known in the physics literature \cite{nestmann2019time,cerrillo2014non,cygorek2025timenonlocalversustimelocallongtime,gasbarri2018recursive} are discussed, highlighting how the expansion in $t$ greatly simplifies the associated calculations. Most importantly, leveraging a few recent results \cite{grigoletto2024exact,grigoletto2025quantummodelreductioncontinuoustime} it is proven that when the original generatos is purely Hamiltonian, the proposed approximation at first and second order provides a CPTP two-parameter semigroup.
    \item In Sec. \ref{sec:applications} the proposed procedure is validated against several prototypical quantum examples: a dephasing spin boson model, a central spin model, and and Ising spin chain. Note that the majority of the applications presented in this work are in the context of quantum physics since, up to this point, this is where these LTV-based model reduction techniques have found most applications.  
\end{itemize}
Throughout the paper, pedagogical examples are included to intuitively explain the most involved proofs. Furthermore, both the strengths and weaknesses of the proposed approximation procedure are highlighted in detail, namely: the recursive formula used to compute the terms of the generators at different orders makes the procedure numerically easier to access at higher orders than the alternatives; allows for an initial model that is in Lindblad form (not only purely Hamiltonian); does not require the model to be put in rotating frame; and does not rely on a weak coupling assumption, thus allowing to study systems in strong coupling regimes. On the other hand, the proposed reduction procedure relies on an expansion in time and thus provides a good approximation only for small times. Although this aspect is limiting, similar restrictions also hold in the well-known approximations of the TCL ME proposed by \cite{shibata_generalized_1977}.

\section{Problem setting and connection with the literature}
\label{sec:problem_setting}
\subsection{Problem setting}

In this work, we assume to be given an LTI stable autonomous model of the type 
\begin{equation}
    \dot{x}_t = L x_t
    \label{eq:original_model_linear}
\end{equation}
where $x_t\in\Cb^n$ is the system's state and $L\in\Cb^{n\times n}$ is the state matrix or generator. Such a generator defines a one-parameter semigroup, i.e. $\{\Lambda_t \equiv e^{L t}\}_{t\geq 0}$ such that $\Lambda_{t+s} =\Lambda_t\Lambda_s$, for all $t,s\geq0$.

We next assume that, for some reason dictated by the application at hand, we are not interested in reproducing all the degrees of freedom contained in the system's state $x_t$, but we are only interested in the degrees of freedom contained on a given subspace $\Vs\subseteq\Cb^n$ of dimension $m=\dim(\Vs)$. Importantly, here we assume the subspace to be {\em given and fixed} $\Vs$.
Let then $P\in\Cb^{n\times n}$, be a projector $P^2=P$ onto $\Vs={\rm Im}P$. Note that $P$ does not need to be orthogonal (with respect to the standard inner product $\langle\cdot,\cdot\rangle$), i.e. $P^\dag \neq P$. For future convenience, let us define $Q\equiv \one-P$ where $\one$ denotes the identity matrix and factorize $P$ into two linear maps: the {\em reduction} $R\in\Cb^{m\times n}$ and the {\em injection} $J\in\Cb^{n\times m}$ such that $P=JR$ and $RJ = \one_m$ the identity in $\Cb^{m\times m}$. 

We further make a simplifying assumption that is quite typical in these settings; see, e.g. \cite{breuer2002theory}: We assume that $x_0\in\Vs$, or equivalently $Px_0 = x_0$.

\begin{figure}
    \centering
    \includegraphics[width=\linewidth]{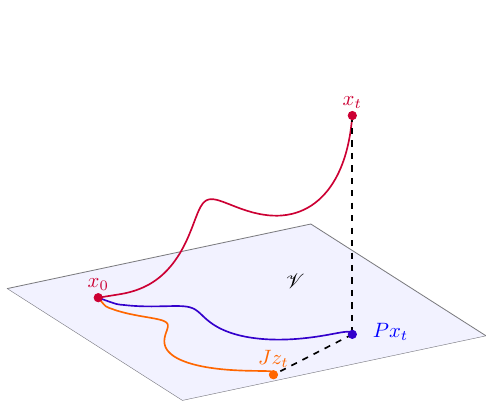}
    \caption{Graphical representation of the approximation problem. The pale blue plane represents the subspace $\Vs$ onto which $P$ projects. The red curve represents the trajectory of the original model $x_t$, which starts in $\Vs$ at $t=0$.  The blue curve represents the projection of the true trajectory $Px_t$, i.e. the one we want to approximate. Finally, the orange curve represents the approximate evolution we aim to find, i.e. $Jz_t$.}
    \label{fig:schematic}
\end{figure}

In summary, we have the model
\[\begin{cases}
\dot{x}_t = L x_{t}\\
y_t = R x_t
\end{cases}, \quad x_0\in\Vs\]
where we defined the output signal $y_t\in\Rb^m$.

In this work, we are interested in deriving a reduced model that is capable of reproducing the trajectories \(y_t = Pe^{Lt}x_0,\quad \forall x_0\in\Vs, \quad \forall t\geq0\)
or, equivalently,
\begin{equation}
    y_t = R e^{Lt} J \check{x}_0,\quad \forall \check{x}_0\in\Cb^m, \quad \forall t\geq0.
\end{equation}

Note that, when $\Vs$ is $L$-invariant, the problem has a trivial solution since by projecting the generator $L$ onto $\Vs$ one obtains the reduced generator $RLJ$ that exactly reproduces all the trajectories of interest.   
One can obtain an exact and minimal LTI reduced model by computing the orthogonal to the non-observable subspace $\Ns^\perp$ as the smallest $L^\dag$-invariant subspace that contains $\Vs$ and projecting the dynamics onto it, thus obtaining the smallest observable LTI model \cite{rosenbrock1970state}. However, the drawback of this approach is the fact that if the couple $(L,R)$ is observable, the model can not be reduced exactly. In those cases, one should recur to approximate model reduction techniques such as balanced truncation \cite{antoulas,gugercin2004survey} and moment matching \cite{astolfi_moment_matching,padoanmomentmatching}.
In those cases, one aims to find a subspace where to project the dynamics in order to obtain an LTI model that approximates the trajectories of interest in a given metric. 

In this work, we take a different route: We fix the subspace where the dynamics is projected to $\Vs$ and lift the time-independence constraint by looking for an LTV model of dimension $m$ that approximates the trajectories of interest. Note that while this approach is not common in the control literature, this has been studied in the physics literature for the derivation of TCL ME as detailed in Sec. \ref{sec:quantum_connection}.

More rigorously, let us define a state $z_t\in\Cb^m$ whose evolution is given by an LTV model of the type
 \begin{equation}
    \begin{cases}
        \dot{z}_t = F_t z_t\\
        z_0 = R x_0
    \end{cases}
 \end{equation}
 where $F_t\in\Cb^{n\times n}$ is a time-dependent generator. The evolution of such an LTV model is given by the two-parameter semigroup $\{\Upsilon_t \equiv \Tb e^{\int_0^t F_s ds} \}_{t\geq0}$ for which the composition rule $\Upsilon_{t,0}=\Upsilon_{t,s}\Upsilon_{s,0}$ for all $0\leq s\leq t$ holds. Here $\Tb$ denotes the time-ordering operator, and $\Tb e$ denotes the time-ordered exponential or, equivalently, the Volterra series \cite{lawrence1979volterra,Giscard_2015}.  

Our aim is then to construct $F_t$ so that $z_t\approx y_t$ or, in other words, we want to determine $F_t$ so that \[\Tb e^{\int_0^t F_s ds} \approx R e^{L t} J.\] The sense in which the approximation $\approx$ is intended shall be defined more precisely later on in the paper. A graphical representation of the approximation problem is given in Fig.\,\ref{fig:schematic}.   

\subsection{Reduced and exact time-dependent dynamics}

Before moving on to the method proposed in this work for computing the time-dependent generator $F_t$ we shall discuss the common derivation found in the physics literature to solve the problem at hand \cite{Tokieda_2025, chaturvedi_time-convolutionless_1979, nestmann_time-convolutionless_2019, shibata_generalized_1977}. The derivation is particularly elegant hence we report it here for the interested reader. 

Dividing the evolution of $x_t$ in the evolution of $P x_t$ and that of $Q x_t$ we have 
\begin{align}
    \frac{d}{dt}P x_t &= P LP x_t + PLQ x_t,\label{eq:P_rho_t}\\
    \frac{d}{dt}Q x_t &= QLP x_t + QLQ x_t.
\end{align}
Solving the second equation for $Q x_t$ with the Lagrange equation, we obtain
\[Qx_t = e^{QLQ t} Qx_0 + \int_0^t e^{QLQ s }QLP x_{t-s} ds.\]
By substituting this equation into Eq. \eqref{eq:P_rho_t} we obtain the famous Nakajima-Zwanzig equation \cite{nakajima,zwanzig}:
\begin{align}
    \frac{d}{dt}P x_t = PLP x_t &+ PL e^{QLQ t} Qx_0 \nonumber\\&+\int_0^t PLQ e^{QLQ s }QLP x_{t-s} ds.
\end{align}
We then recall that 
\(x_{t-s} = e^{-L s}x_t\)
so that 
\[Qx_t = e^{QLQ t} Qx_0 +\int_0^t e^{QLQ s }QLP e^{-L s} ds\, \cdot\, (P + Q) x_t. \]
Now, defining 
\begin{equation}
    M_t \equiv \one - \int_0^t e^{QLQ s }QLP e^{-L s} ds
    \label{eq:M_t}
\end{equation}
we have 
\[M_t Qx_t = e^{QLQ t} Qx_0 + (\one-M_t) Px_t.\]
If we then assume that $M_t$ is invertible, we obtain 
\[Qx_t = M_t^{-1}e^{QLQ t} Qx_0 + (M_t^{-1} -\one) P x_t.\]
Substituting this into Eq. \eqref{eq:P_rho_t} we obtain 
\begin{align}
    \frac{d}{dt}Px_t &= PLP x_t + PLM_t^{-1}e^{QLQ t} Qx_0 \nonumber \\& \qquad\qquad\qquad\qquad+ PL (M_t^{-1}-\one) P x_t \nonumber\\ &= PLM_t^{-1}e^{QLQ t} Qx_0 + PLM_t^{-1} Px_t.
    \label{eq:TCL_ME}
\end{align}
Note that under the assumption $Px_0=x_0$ we have $Qx_0=0$ and thus Eq. \eqref{eq:TCL_ME} gets simplified to 
\begin{equation}
  \frac{d}{dt}Px_t = \underbrace{PLM_t^{-1} P}_{\equiv K_t} \rho_t  
  \label{eq:tcl_me_no_affine}
\end{equation}
which, separating the projector $P$ into its factors $P = JR$, becomes  
\begin{equation}
    \frac{d}{dt}\check{x}_t = \underbrace{RL M_t^{-1} J}_{\equiv F_t} \check{x}_t.
    \label{eq:tcl_me_reduced}
\end{equation}

A few comments on this solution are in order:\\ 
i) The invertibility of $M_t$ needs to be verified and typically only holds fot small $t$, (see, e.g. \cite{breuer2002theory}) but for the times where $M_t$ is invertible, Eq. \eqref{eq:tcl_me_reduced} provides an exact solution to the problem at hand. \\
ii) Note that, to compute $F_t$ (or equivalently $K_t$) one needs to first compute the integral that defines $M_t$ and then compute the inverse of $M_t$. Both these tasks are challenging (numerically and analytically) and it is generally considered that solving this integral and inverse is as difficult as solving the original model \eqref{eq:original_model_linear}, see e.g. \cite{breuer2007}. However, Eq. \eqref{eq:TCL_ME} provides an excellent starting point for computing an approximation of the reduced evolution. 

To avoid these numerical issues, $K_t$ is typically approximated by expanding it as an expansion in the coupling strength $\varepsilon$ and by computing the terms of the expansion as ordered cumulants \cite{breuer2002theory,breuer2006,nestmann2019time}.
In many cases of interest, this approximation provides a good approximation of the open quantum system dynamics, both in the long-time and short-time regimes. It is, however, important to note that there is no guarantee that such an approximation holds also for long times. As an example of this, see \cite{breuer2007} where it is shown that TCL2 and TCL4 (approximations at orders $N=2,4$) do not approximate the true dynamics in the long-time regime.
We refer the reader to Sec. \ref{sec:connection_with_literature} for further connections to the physics literature.

\section{Polynomial approximations of the time-dependent generator}
\label{sec:time_ordered}

In this section we derive the key result of this work: a method to approximate the time-dependent generator defined in Eq.\eqref{eq:tcl_me_reduced}, $F_t$. The derivation that follows is based on the following observation. 
\begin{proposition}
    Let $M_t$ be the time-dependent operator defined in Eq. \eqref{eq:M_t} and let $F_t$ be the time-dependent generator defined in Eq. \eqref{eq:tcl_me_reduced}. Then, for every $t$ such that $M_t$ is invertible, $F_t$ is analytic in $t$, i.e. 
    \[F_t = \sum_{k=0}^{+\infty} t^k F_{(k+1)}.\]
\end{proposition}
\begin{proof}
    We can start by observing that the function $e^{QLQt}QLPe^{-Lt}$ is the product of two analytic functions in $t$ hence it is analytic in $t$. Furthermore, as the integral of an analytic function is also analytic, see, e.g. \cite{bourchtein2021analytic} we have that $M_t$  is also analytic in $t$. Finally, as the inverse of an analytic function $f(t)$ is analytic where $f(t)\neq0$, we have that $F_t$ is also analytic in $t$ whenever $M_t$ is invertible. 
\end{proof}
Note that we label the coefficients in the power series starting from $1$ instead of $0$ for later convenience. This insightful result allows us to: 1) focus on computing the terms of the power series $\{F_{(k+1)}\}_{k=0}^{+\infty}$ and 2) approximate $F_t$ by truncating the infinite power series to a finite order $N$. In what follows, we thus approximate $F_t$ with the {\em polynomial time-dependent generator} 
\[F_{t,N} \equiv \sum_{k=0}^N t^k F_{(k+1)}.\]
Other than an obvious computational benefit introduced by computing the terms $F_{(k)}$ only to order $N$, the assumption that the generator is polynomial in the time variable greatly simplifies the computation of the time-ordered exponential $\Tb e^{\int_0^t F_{s,N} ds}$ as shown in the next key result. 

\begin{theorem}
\label{thm:time_ordered_expansion}
    Let us assume that \(F_{t,N} = \sum_{k=0}^{N} t^k F_{(k+1)}\). 
    Then the time-ordered exponential can be written as
    \begin{equation}
        \Tb e^{\int_0^t F_{s,N} ds} = \sum_{k=0}^{+\infty} t^k E_{(k)}
        \label{eqn:time_ordered_expansion}
    \end{equation}
    where $E_{(0)}=\one$ and the other terms are defined by the recursion:
    \begin{equation}
        E_{(k)} =  \frac{1}{k} \sum_{s=1}^{k} F_{(s)}E_{(k-s)}
        \label{eqn:E_k}
    \end{equation}
    and where we used the convention $F_{(k)}=0$ for all $k>N$.
\end{theorem}
Proof of this result and an explanatory example are given in Appendix \ref{sec:thechnical_results:time_ordered}.

Note that Eq. \eqref{eqn:E_k} provides a closed formula for the coefficients $E_{(k)}$ of the expansion of $\Tb e^{\int_0^t F_{s,N} ds}$ as a function of the coefficients $F_{(k)}$. Furthermore, only the terms $\{E_{(\ell)}\}_{\ell=1}^{k-1}$ appear in Eq. \eqref{eqn:E_k}, which means that this equation can be used iteratively to compute each $E_{(k)}$. 

Theorem \ref{thm:time_ordered_expansion} provides the main intuition on how to compute the coefficients $\{F_{(k)}\}_{k=1}^{N+1}$ of the polynomial time-dependent generator $F_{t,N}$: Eq. \eqref{eqn:time_ordered_expansion} provides an expansion for $\Tb e^{\int_0^t F_{s,N} ds}$ in the time variable $t$. We can then expand $e^{tL}$ in $t$ and construct $F_{(k)}$ so that the terms of $\Tb e^{\int_0^t F_{s,N} ds}$ match the terms of $R e^{tL} J$ by orders of $t$. This approach is similar in principle to the adiabatic elimination techniques \cite{adiabadic-elimination,Tokieda_2025} where, instead of matching orders of $\varepsilon$, we match orders of $t$. The resulting procedure, which represents the main contribution of this work, is stated next.

\begin{figure*}[th]
    \centering
    \begin{subfigure}{0.48\textwidth}
        \includegraphics[width=\linewidth]{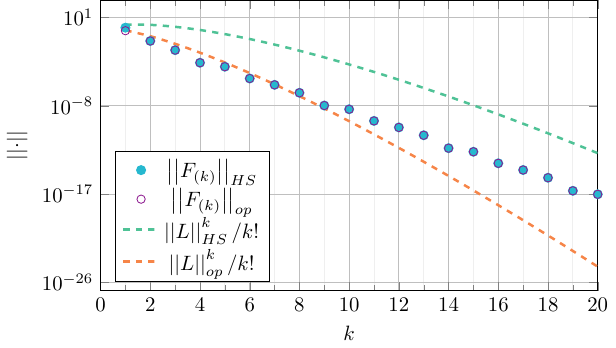}
    \end{subfigure}
    \begin{subfigure}{0.48\textwidth}
        \includegraphics[width=\linewidth]{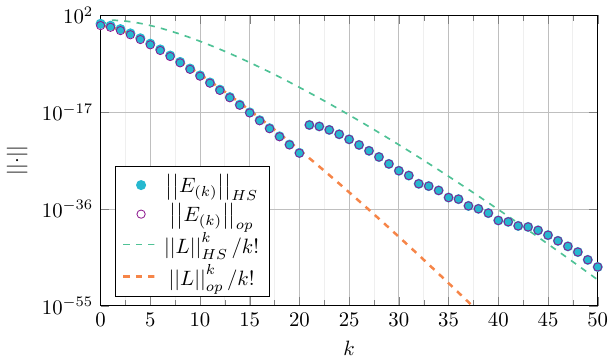}
    \end{subfigure}
    \caption{Operator and Hilbert-Schmidt norm of the first 20 $F_{(k)}$ terms (left) and of the first 50 $E_{(k)}$ terms when $N=20$ (right).}
    \label{fig:norm_bound}
\end{figure*}

\begin{theorem}
\label{thm:approx}
Let us define 
\begin{equation}
F_{t,N} \equiv \sum_{k=0}^{N} t^k F_{(k+1)}
\label{eqn:poly_generator}
\end{equation}
where 
\begin{equation}
F_{(k)} \equiv k\frac{R L^k J}{k!} - \sum_{h=1}^{k-1}F_{(k-h)}\frac{R L^h J}{h!}.
\label{eqn:F_k}
\end{equation}
Then:
\begin{align}
R e^{L t}J - \Tb e^{\int_0^t F_{s,N} ds} = O(t^{N+1}).
\label{eqn:time_approx}
\end{align}
\end{theorem}
\begin{proof}
First of all we can use Theorem \ref{thm:time_ordered_expansion} to write 
\[\Tb e^{\int_0^t F_{s,N} ds} = \sum_{f=0}^{+\infty} t^k E_{(k)}\]
where \(E_{(k)} =  \frac{1}{k} \sum_{s=1}^{k} F_{(s)}E_{(k-s)}= \frac{1}{k}\sum_{h=0}^{k-1} F_{(k-h)}E_{(h)}\) and where we reversed the order of the summation by imposing $h=k-s$. We next prove that, for $1\leq k \leq N$, we have $E_{(k)}=\frac{RL^kJ}{k!}$ by induction. The case $k=1$ is trivial as $E_{(1)} = F_{(1)}E_{(0)} = RLJ$. Let then assume that $E_{(h)} = \frac{RL^hJ}{h!}$ for all $h<k$ and let us prove that $E_{(k)}=\frac{RL^k J}{k!}$. Then
\begin{align*}
    E_{(k)} &= \frac{1}{k}\sum_{h=0}^{k-1} F_{(k-h)}E_{(h)} = \frac{F_{(k)}}{k} + \frac{1}{k}\sum_{h=1}^{k-1}F_{(k-h)}E_{(h)}\\
    &= \frac{F_{(k)}}{k} + \frac{1}{k}\sum_{h=1}^{k-1}F_{(k-h)}\frac{RL^h J}{h!}
    = \frac{R L^k J}{k!}
\end{align*}
where we simply substituted the definition of $F_{(k)}$ given in eq. \eqref{eqn:F_k}.

By leveraging the definition of the exponential map $e^{L t}$ we can expand the term $R e^{L t}J$ to obtain
\begin{align*}
    R e^{L t} J = \sum_{k=0}^{+\infty} \frac{t^k}{k!} RL^k J.
\end{align*}
Combining the two terms we obtain
\begin{align*}
    R e^{L t}J - \Tb e^{\int_0^t F_{s,N} ds} &= \sum_{k=0}^{+\infty}t^k\left(\frac{RL^k J}{k!} - E_{(k)}\right) \\
    &= \sum_{k=0}^{N}t^k\cancel{\left(\frac{RL^k J}{k!} - E_{(k)}\right)}\\ &\qquad\qquad + O(t^{N+1})
\end{align*}
which concludes the proof.
\end{proof}

We shall observe a few key aspects. 
First, Eq. \eqref{eqn:time_approx} finally gives sense to the approximation symbol $\approx$ used in Sec. \ref{sec:problem_setting} i.e. an approximation of the type $O(t^{N+1})$. This implies that the approximation holds only for small times. 
Second, it is important to notice that the right-hand side of equation \eqref{eqn:F_k} depends on $RL^{k}J$ and $\{F_{(s)}\}_{s=1}^{k-1}$. This means that equation \eqref{eqn:F_k} can be used iteratively to construct the terms $F_{(k)}$ starting from $k=1$ and growing up to $k=N$. 
For the reader's convenience the first few terms of the series are reported next:
\begin{align*}
    F_{(1)} &= RLJ\\
    F_{(2)} &= RL^2J - F_{(1)}RLJ = RL^2J - (RLJ)^2\\
    F_{(3)} &= \frac{RL^3J}{2} - F_{(2)}RLJ - F_{(1)}\frac{RL^2J}{2} \\
    &= \frac{RL^3J}{2} - \frac{F_{(1)}F_{(2)}}{2} - \frac{F_{(1)}^3}{2} - F_{(2)}F_{(1)}
\end{align*}
It is interesting to notice a couple of points about the first few terms: i) As one would expect, the term of order one $F_{(1)}$ coincides with the reduction of the original dynamics which is often used as rough time-independent approximation of the reduced evolution \cite{antoulas}; ii) The second order term does not only contain the reduction of the generator squared $RL^2J$ but also contains a correction term depends on the first order term $F_{(1)}$.

Lastly, it is important to note that similar (yet distinct) recursive formulas to derive similar approximations have been found in the physics literature \cite{nestmann2019time,cerrillo2014non,cygorek2025timenonlocalversustimelocallongtime,gasbarri2018recursive}.  We refer the reader to Sec.\ref{sec:connection_with_literature} where we discuss connections and differences between these works.

\subsection{Linear testbed}
\label{sec:linear_testbed}

To showcase the capability of the proposed method, we here consider a simple numerical example.  
Specifically, we consider a random autonomous and stable system of size $n=20$. 
Let then $A$ be a random uniform matrix in $\Cb^{n\times n}$, i.e. $[A]_{i,j}\sim\Uc([0,1])$ and let $B = A - \lambda_{\rm Max} \one_{n}$ where $\lambda_{\rm Max}$ is the absolute value of largest eigenvalue of $A$. 
The dynamic's generator $L\in\Cb^{n\times n}$ is then obtained as $L = \frac{1}{2}\frac{B}{\norm{B}_{op}}$ where $\norm{\cdot}_{op}$ denotes the operator norm.
This ensures that $L$ is simply stable and has operator norm $\norm{L}_{op}=\um$. 

The projector $P$ is then constructed through $R = \left[\begin{array}{c|c} \one_{m\times m} & 0_{m\times n-m} \end{array}\right]$ and $J = R^\dag$ to obtain \[P = \left[\begin{array}{c|c} \one_{m\times m} & 0_{m\times n-m} \\\hline 0_{m \times n-m } & 0_{n-m\times n-m} \end{array}\right].\]

By using Eq. \eqref{eqn:F_k} we can then compute the terms $F_{(k)}$ iteratively for $k=1,\dots,N$. Fig. \ref{fig:norm_bound} (left) shows both the Hilbert-Schmidt and operator norm of the terms $F_{(k)}$ for $k=1,20$. Interestingly, the Hilbert-Schmidt norm, $\norm{F_{(k)}}_{HS}$, is bounded by $\frac{\norm{L}_{HS}^k}{k!}$, indicating that the series $F_{t,N}$ should be converging for $N\to\infty$. Curiously, the same does not hold for the operator norm.
Fig. \ref{fig:norm_bound} (right) instead shows the norm of the terms $E_{(k)}$ where we fixed $N=20$ and $F_{(k)}=0$ $\forall k>20$. Interestingly, for $k\leq N$ $\norm{E_{(k)}}$ is bounded by $\frac{\norm{L}^k}{k!}$ while for $k>N$ this is no longer the case.  

\begin{figure}[th]
    \centering
    \includegraphics[width=\linewidth]{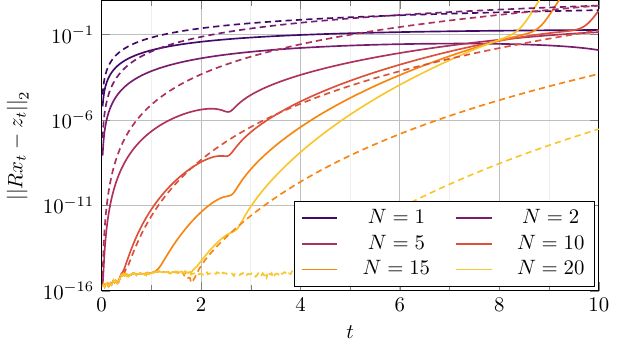}
    \caption{The continuous lines represent the error norm between the exact reduced evolution $R x_t$ and the approximated reduced evolution $z_t$ for different values of the approximation order $N$. The dashed curves represent an approximation obtained by truncating the Taylor expansion of the exponential for different values of the truncation order $N$. }
    \label{fig:error_linear}
\end{figure}

Fig. \ref{fig:error_linear} shows the error norm $\norm{Rx_t-z_t}_2$ versus time. In this simulation $Rx_t$ was simulated by reducing the original model, i.e. $Rx_t= R e^{L t}x_0$ while $z_t$ is obtained by simulating the time-dependent linear differential equation $\dot{z}_t=F_{t,N}z_t$, for different values of $N$, by using an approximation of the equation given in Theorem \ref{thm:time_ordered_expansion}, i.e. \(\Tb e^{\int_0^t F_{s,N} ds} \approx \sum_{k=0}^{100} t^k E_{(k)}\).  
It can be seen that increasing the degree of approximation $N$ decreases the error committed by the approximated model. 

An alternative natural numerical approximation that comes to mind is given by truncating the exponential expansion at a given order $N$, i.e. $Rx_t\approx \sum_{k=0}^{N} \frac{t^k}{k!} RL^k J\, Rx_0$. Note that such an approximation has numerical complexity almost identical to that of the approximation proposed in this work, as in both cases one needs to compute $L^N$. Furthermore, from a theoretical point of view, both cases commit an error of the order of $O(t^{N+1})$. To test this fact numerically, in Fig. \ref{fig:error_linear} we included the actual error committed by these two approximations at the same orders of approximation: the continuous line denotes the error committed by the procedure proposed here, while the dashed line is the error committed by the truncation of the exponential. One can observe that for lower orders ($N=1,2,5$) the procedure proposed here yields better results, whereas at higher orders this is not the case. 

\section{Ensuring valid quantum models at lower orders}
\label{sec:cptp_second_order}

In this section, we apply the results derived in the Sec. \ref{sec:time_ordered} to quantum dynamical models, the application that originally inspired this work. 

\subsection{The reduction problem for quantum dynamical models}
\label{sec:quantum_connection}
Let us consider a quantum system defined on a Hilbert space $\Hc$, and let $\Bf(\Hc)$ denote the set of bounded operators acting on $\Hc$. Unless differently specified, We only consider finite-dimensional systems, i.e. $\dim(\Hc)<\infty$, $\Hc\simeq\Cb^n$ and $\Bf(\Hc)\simeq \Cb^{n\times n}.$
The state of such a system is described by a density operator $\rho\in\Df(\Hc)$ with $\Df(\Hc) = \{\rho\in\Bf(\Hc)|\, \rho=\rho^\dag\geq0,\,\tr\rho=1\}$. Here, we assume that the state of the system evolves according to the Von-Neumann master equation 
\begin{equation}
    \dot{\rho}_t = \Lc(\rho_t)
    \label{eq:original_model}
\end{equation}
where $\rho_0\in\Df(\Hc)$ is the system's initial condition, and $\Lc$ is the celebrated Lindblad (or GKLS) generator \cite{lindblad1976generators,Gorini:1975nb}: 
\begin{equation}
    \Lc(\rho) = -i[H,\rho] + \sum_k\Dc_{L_k}(\rho)
    \label{eqn:lindblad_generator}
\end{equation}
$H=H^\dag$ is the system's Hamiltonian and $\Dc_L(\rho) = L\rho L^\dag - \um\{L^\dag L, \rho\}$ represents the system's dissipation toward a Markovian environment represented by a set of operators $\{L_k\}$ called noise operators. The evolution of the quantum system is then described by the exponential map $\rho_t = e^{\Lc t}(\rho_0)$. Importantly, $e^{\Lc t}$ generates a quantum dynamical semigroup, i.e. $\{\Lambda_t \equiv e^{\Lc t}\}_{t\geq 0}$ a one-parameter semigroup of CPTP maps also known as Markovian dynamics.
In the following, we will say that a super-operator is of \textit{Lindblad-type} if it can be put in the form of Eq. \eqref{eqn:lindblad_generator}.

Note that in general, in the derivation of TCL ME \cite{shibata_generalized_1977,chaturvedi_time-convolutionless_1979}, the initial ‘‘big'' model is considered to be Hamiltonian, i.e. $\Lc(\cdot)=-i[H,\cdot]$. However, mathematically, there is no reason to restrict to this case; hence, in this work, we start with the general assumption that $\Lc$ is a generator of Lindblad form and then focus on the Hamiltonian case when necessary.

In many situations of interest in quantum physics, relevant projectors are represented by CPTP ones. We thus assume to be given a CPTP, not necessarily orthogonal (w.r.t. $\inner{\cdot}{\cdot}_{HS}$), projector $\Pc:\Bf(\Hc)\to\Bf(\Hc)$, $\Pc^2=\Pc$, whose image will be denoted by $\Vs \equiv {\rm Im}\Pc$. This provides a generalization to the superoperator projector setting considered by \cite{breuer2006,breuer2007}.  Furthermore, for the sake of simplicity, we assume that $\Pc$ admits\footnote{Note that this assumption is made for simplicity of presentation and can easily be lifted by restricting onto the support of ${\rm Im}\Pc$, see e.g. \cite{wolf2012quantum}.} a full rank fixed point, i.e. $\exists\rho\in\Df(\Hc)$, $\rho>0$ such that $\Pc(\rho)=\rho$. 
As we shall see in more detail later (in Sec.\ref{sec:cptp_first_order}) such a projector is, in fact, a conditional expectation in quantum probability theory \cite{petz2007quantum,blackadar2006operator} and enjoys some useful and well known properties \cite{lindblad1999general, wolf2012quantum, ticozzi2017alternating, bialonczyk2018application, blume2010information,grigoletto2024exact,wolf2010inverseeigenvalueproblemquantum}. However, for now it is sufficient to know that $\Pc$ can be decomposed into two factors $\Pc = \Jc\Rc$ which take the name of reduction and injection maps $\Rc:\Bf(\Hc)\to\Bf(\Kc)$ and $\Jc:\Bf(\Kc)\to\Bf(\Hc)$ such that $\Rc\Jc = \Ic_\Kc$ the identity superoperator, ${\rm Im}\Jc = \Vs$ ${\rm Im}\Rc=\As$ and where $\As\subseteq\Bf(\Kc)$ is a $*$-subalgebra of $\Bf(\Kc)$ for some Hilbert space $\Kc$, $\dim(\Kc)\leq\dim(\Hc)$.

To provide some intuition about these mathematical objects, we anticipate the reduction and injection maps we are going to use in the applications of Sec.\,\ref{sec:applications}:
\begin{itemize}
    \item Assume a bipartite case, i.e. $\Hc=\Hc_S\otimes\Hc_B$. We can then consider $\Pc(\rho) = \tr_B(\rho)\otimes\tau$ for some $\tau\in\Df(\Hc_B)$, $\tau>0$ to be our CPTP projector. Then, by chosing $\Kc=\Hc_S$ and $\As=\Bf(\Hc_S)$ we have the maps $\Rc(\rho) = \tr_B(\rho)$ and $\Jc(\mu) = \mu\otimes\tau$.
    \item Another CPTP projector we might want to consider is the projector onto the diagonal $\Pc(\rho)= \sum_{j=0}^{n-1} \ketbra{j}{j} \rho \ketbra{j}{j}$, in which case $\Kc=\Hc$, $\As$ is the commutative algebra given by $\As\equiv\Span\{\ketbra{j}{j},\, j=0,\dots,n-1\}\simeq\Cb^n$. $\Rc(\rho) = \sum_{j=0}^{n-1} \ketbra{j}{j}\rho \ket{j}$ so that $\mu=\Rc(\rho)\in\Cb^n$, and $\Jc(\mu)=\sum_{j=0}^{n-1} \ketbra{j}{j} \bra{j}\mu$ factorize $\Pc$.
\end{itemize}

The reduction map $\Rc$ determines the reduced state, i.e. the one we want to keep track of. Let $\eta_t\in\Df(\Kc)$ be defined as $\eta_t = \Rc(\rho_t)$ for all $t\geq0$. Under the assumption $\Pc(\rho_0)=\rho_0$,
the evolution of the reduced state is thus given by 
\begin{align*}
    \eta_t = \Rc e^{\Lc t} (\rho_0) = \Rc e^{\Lc t}\Jc (\eta_0).
\end{align*}

Theorem \ref{thm:approx} then offers a method to approximate this evolution with a time-dependent generator $\Fc_{t,N}$ that is polynomial in time $t$, i.e. $\Fc_{t,N} = \sum_{k=0}^N t^k \Fc_{(k+1)}$ and such that   
\[\norm{\Tb e^{\int_0^t \Fc_{s,N} ds} - \Rc e^{\Lc t} \Jc} = O(t^{N+1}).\]

Desirable properties for $\Fc_{t,N}$ certainly include the fact that we would like the semigroup $\{\Upsilon_t \equiv \Tb e^{\int_0^t \Fc_s ds} \}_{t\geq0}$ to be a set of CPTP maps for all $t\geq0$. Note that necessary conditions to impose on $\Fc_t$ so that $\Tb e^{\int_0^t \Fc_s ds}$ is CPTP for all $t\geq0$ are not known \cite{benatti2024quantum}. Nevertheless, in the following subsections, we shall see that, under mild assumptions, the first two orders of the approximation ensure that this desirable property is preserved.

\subsection{Connection with the physics literature}
\label{sec:connection_with_literature}

In the typical setting of TCL ME derivation, the system is considered to be bipartite, i.e. $\Hc=\Hc_S\otimes\Hc_B$ and the dynamics purely Hamiltonian, i.e. $\Lc(\rho)=-i[H,\rho]$ with Hamiltonian $H = H_0 + \varepsilon V$ where $H_0$ is the free Hamiltonian of the system and bath, $V$ is the interaction Hamiltonian and $\varepsilon$ is the coupling strength.
Under these assumptions, and assuming that the coupling strength is small, the time-dependent generator $\Kc_t$ defined in Eq. \eqref{eq:tcl_me_no_affine} is computed by performing a perturbative expansion in the coupling parameter. The time-dependent generator is then expressed as 
\begin{align*}
    \Kc_t = \sum_{n=1}^\infty \Kc_n(t)
\end{align*}
where
\begin{align*}
    \Kc_n(t) &= \int_{0}^t dt_1 \int_0^{t_1} dt_2 \dots \int_0^{t_{n-2}}dt_{n-1}\\
    &\qquad \times \sum_q (-1)^q \Pc \Lc_{t} \dots \Lc_{t_i}\Pc\Lc_{t_j}\dots\Lc_{t_k}\Pc\dots\Pc
\end{align*}
and where $\Lc_t = -i[H(t),t]$ and $H(t)$ is the Hamiltonian in rotating frame with respect to $H_0$. The series is then approximated by stopping at a given order $N$, i.e. $\Kc_t\approx\sum_{n=1}^N\Kc_n(t)$.

Under these assumptions, a recursive formula similar to Eq. \eqref{eqn:F_k} has been derived.  
Specifically \cite[Eq. (38)]{nestmann2019time} (but also \cite[Eq. (2)]{cerrillo2014non} \cite[Eq. (2)]{cygorek2025timenonlocalversustimelocallongtime}, \cite[Eq. (16)]{gasbarri2018recursive}) provides a recursive formula for the terms $\Kc_n(t)$ of the time-dependent generator $\Kc_t$. However, these recursive formulas were derived in the rotating frame (w.r.t. $H_0$) under a weak coupling assumption and thus from an initial time-dependent model. This implies that the terms of the recursive formula contain time-ordered integrals that need to be solved to compute the terms of the expansion. 

In contrast, in this work, we start from a time-independent model, hence Eq. \eqref{eqn:F_k} (which was derived independently from \cite{nestmann2019time,cerrillo2014non,cygorek2025timenonlocalversustimelocallongtime,gasbarri2018recursive}) does not contain integrals, resulting in simpler computations.

Furthermore, the main difference that separates this work from most of the physics literature, see e.g. \cite{colla2025unveiling,colla2025recursive,Tokieda_2025} is the fact that, instead of expanding the time-dependent generator in a power series of a perturbation parameter $\varepsilon$, we expand in the time variable $t$. This difference allows for a few benefits: It simplifies the calculations of the approximations and lifts the weak-coupling assumption, thus allowing one to explore the effects of strong-coupling regimes. 

The main drawback introduced by the expansion in the time variable $t$ is the fact that the approximate model holds only for small times. Nevertheless, as we highlighted above, while this consequence is clearly limiting, this situation is comparable to the typical derivation of TCL ME since no guarantees on the long-time limit are available. In contrast, in this work we were able to characterize how the error between the original and reduced models grows with $t$.  

\subsection{Complete positivity of the first order}
\label{sec:cptp_first_order}

In this section, we recall relevant known results about CPTP projectors, their images, and reduced dynamics onto distorted algebras from the literature \cite{lindblad1999general, wolf2012quantum, ticozzi2017alternating, bialonczyk2018application, blume2010information,grigoletto2024exact,wolf2010inverseeigenvalueproblemquantum}. The consequences of the assumptions we made about the projector $\Pc$ are summarized in the following proposition. 
\begin{proposition}
\label{prop:structure_of_fixed_points}
    Let $\Pc$ be a CPTP projector with a full rank fixed point $\bar{\rho}>0$, $\Pc(\bar{\rho})=\bar{\rho}$ and let $\Vs$ be its image, i.e. $ \Vs \equiv {\rm Im}\Pc \subseteq\Bf(\Hc)$. 
    Then there exist a decomposition of the Hilbert space \[\Hc = \bigoplus_k (\Hc_{F,k}\otimes\Hc_{G,k}),\] a unitary operator $U\in\Bf(\Hc)$, and a set of density operators $\tau_k\in\Df(\Hc_{G,k})$ such that 
    \begin{equation}
        \Vs = U\left(\bigoplus_k \Bf(\Hc_{F,k}) \otimes \tau_k  \right) U^\dag.
        \label{eq:wedderburn_decomposition}
    \end{equation}
    Furthermore, $\Vs$ is a $*$-algebra closed with respect to the product $X\cdot_\sigma Y \equiv X \sigma Y$ where \[\sigma = U\left(\bigoplus_k\one_{F,k}\otimes \tau_k^{-1}\right)U^\dag.\]  
    Let then $W_k$ be a linear operator from $\Hc$ to $\Hc_{F,k}\otimes\Hc_{G,k}$ such that $W_k W_k^\dag=\one_{F,k}\otimes \one_{G,k}$ \cite{wolf2012quantum}. Then the CPTP projection $\Pc$ onto $\Vs$ takes the form
    \begin{align}
         \Pc(X) 
            &= U\left[ \bigoplus_{k} \tr_{\Hc_{G,k}}\left(W_k X W_k^\dag \right)\otimes \tau_k  \right] U^\dag.
        \label{eqn:state_ext_blocks}
    \end{align}
    and is orthogonal with respect to the inner product $\inner{X}{Y}_{\sigma,\lambda} = \tr[X^\dag \sigma^\lambda Y \sigma^{1-\lambda}]$ for $\lambda\in[0,1]$.
\end{proposition}
Proof of this result can be found in \cite{lindblad1999general} and \cite{wolf2010inverseeigenvalueproblemquantum,wolf2012quantum,tit2023, ticozzi2017alternating, johnson2015general}. Note that projectors of this type take the name of state extensions and are the dual of conditional expectations \cite{petz2007quantum,blackadar2006operator}. Furthermore, note that a $*$-algebra closed with respect to a modified product often takes the name of \textit{distorted algebra} \cite{blume2010information}, while the block structure provided in Eq. \eqref{eq:wedderburn_decomposition} takes the name of \textit{Wedderburn decomposition}. 

The previous result provides us with a structure for the image of a  CPTP projection whenever there exists a full rank state $\bar{\rho}$ such that $\Pc(\bar{\rho}) = \bar{\rho}$. From this structure one can observe that a distorted algebra $\Vs$ is isomorphic to a smaller representation which we denote $\As$ where the weighted copies are removed, i.e. $\As \equiv \bigoplus_k \Bf(\Hc_{F,k})$. This allows us to find the two factors $\Rc:\Bf(\Hc)\to{\As}$ and $\Jc:{\As}\to\Bf(\Hc)$ such that, $\Rc\Jc = \Ic_{{\As}}$ and $\Pc = \Jc\Rc$. 
Furthermore, one can prove \cite{grigoletto2024exact} that these two factors can be chosen to be CPTP, i.e. for all $X\in\Bf(\Hc)$
\begin{equation}
\label{eqn:reduction}
\Rc(X)=\bigoplus_k \tr_{\Hc_{G,k}}(W_k X W_k^\dag)=\bigoplus_k X_{F,k}=\check X ,
\end{equation}
and for all $\check{X}\in\check{\As}$, $\check{X}=\bigoplus_k X_{F,k}$,
\begin{equation}
\label{eqn:injection}
\es(\check X)=U\bigg(\bigoplus_k X_{F,k} \otimes\tau_{k}\bigg)U^\dag.
\end{equation}

The reduction and injection maps given in equations \eqref{eqn:reduction} and \eqref{eqn:injection} then allow us to compute the reduced dynamics $ \Rc\Lc\Jc$ restricted to the algebra ${\As}$. Then we can prove that the reduced generator $\check{\Lc}$ is GKLS by leveraging the structure of the reduction and injection maps. Here we recall \cite[Theorem 4]{grigoletto2024exact} for completeness.
\begin{theorem}
\label{thm:Lindblad_reduction}
Let $\Pc$ be a CPTP projection with a full rank fixed point, let $\mathscr{V}$ be its image $\Vs = {\rm Im }\Pc \subseteq \mathcal{B(H)}$, and let $\Rc$ and $\Jc$ denote the CPTP factorization of $\Pc = \Jc\Rc$, as defined in Eq.\,\eqref{eqn:reduction} and \eqref{eqn:injection}. Then for any Lindblad generator $\Lc$, its restriction to ${\As}$, 
$\Rc\Lc\Jc,$ 
is also a Lindblad generator, that is, $\Rc\Lc\Jc:\As\to\As$ and $\{e^{\check{\Lc}t}\}_{t\geq 0}$ is a quantum dynamical semigroup. 
\end{theorem}
Note that this theorem only proves that $\Rc\Lc\Jc$ is of Lindblad type. In order to compute the reduced Hamiltonian and noise operators, one can use the procedure described in \cite[Appendix]{grigoletto2025quantummodelreductioncontinuoustime}.
This result allows us to state that, under the assumptions described above, the approximation at first order $\Fc_{t,1} = \Rc\Lc\Jc$ is of Lindblad type and thus it is guaranteed to give a CPTP semigroup for all times.    

\subsection{Second order: The case of purely Hamiltonian evolution}

While the results we derived up to this point hold for general Lindblad generators $\Lc$, we shall now focus on the special case of purely Hamiltonian evolution, i.e. $\Lc(\cdot) = -i[H,\cdot].$ Other than the fact that this assumption is interesting in practical cases, we shall see that under this simplifying assumption we are able to prove that the second order term $\Fc_{(2)}$ is also of Lindblad-type.

The next result is based on the fact that the squared of a purely Hamiltonian generator is equal to a dissipator:
\begin{align}
    \Lc^2(\rho) &= - [H[H,\rho]] = -(H^2\rho - 2H\rho H + \rho H^2)\nonumber \\
    &= 2H\rho H - \{H^2,\rho\} = \Dc_{\sqrt{2}H}(\rho).
    \label{eqn:square_hamiltonian}
\end{align}
This fact allows us to prove the next result.
\begin{theorem}
\label{thm:cptp_second_order}
    Under the assumptions of Theorem \ref{thm:Lindblad_reduction} we have that $\Fc_{(2)} = \Rc\Lc^2\Jc - (\Rc \Lc \Jc)^2$ is of Lindblad-type and for $\Fc_{t,2} \equiv \Fc_{(1)} + t \Fc_{(2)}$ we have that $\Tb e^{\int_0^t \Fc_{s,2} ds}$ forms a CPTP semigroup for all $t\geq0$. 
\end{theorem}
\begin{proof}
    By leveraging \cite[Proposition 5]{grigoletto2024exact} we have that $\Rc\Lc\Jc(\check{\rho}) = -i[\check{H},\check{\rho}]$ where $\check{H} = \Jc^\dag(H)$ hence, using equation \eqref{eqn:square_hamiltonian} we obtain $(\Rc\Lc\Jc)^2(\check{\rho}) = \Dc_{\sqrt{2}\check{H}}(\check{\rho})$. On the other hand, by leveraging \cite[Proposition 6 and Proposition 7]{grigoletto2025quantummodelreductioncontinuoustime}, we have that $\Rc \Dc_{\sqrt{2}H}\Jc(\check{\rho}) = \sum_{\check{C}\in\Xc(\sqrt{2}H)} \Dc_{\check{C}}(\check{\rho})$ for a set of operators $\Xc(\sqrt{2}H)\subset \Bf(\Kc)$. Importantly, by \cite[Corollary 4]{grigoletto2025quantummodelreductioncontinuoustime}, we have that $\Jc^\dag(H)\in\Xc(\sqrt{2}H)$ and thus we can write $\Rc \Dc_{\sqrt{2}H}\Jc(\check{\rho}) = \sum_{\check{C}\in\Xc(\sqrt{2}H)/\{\sqrt{2}\check{H}\}} \Dc_{\check{C}}(\check{\rho}) + \Dc_{\sqrt{2}\check{H}}(\check{\rho})$ since $\sqrt{2}\check{H} = \Jc^\dag(\sqrt{2}H)$. Then, by combining these two results we obtain 
    \[ \Fc_{(2)} = \Rc\Lc^2\Jc - (\Rc \Lc \Jc)^2 = \sum_{\check{C}\in\Xc(\sqrt{2}H)/\{\sqrt{2}\check{H}\}} \Dc_{\check{C}}(\check{\rho}) \]
    which is the sum of a set of GKLS, purely dissipative generators and thus of Lindblad type. Combining this with Theorem \ref{thm:Lindblad_reduction} it trivially follows that $\Fc_{t,2}$ is of Lindblad type for all $t\geq0$.
\end{proof}
\subsection{Second order in the special case of bipartite systems}
\label{sec:second_order_bipartite}
Since the proof of Theorem \ref{thm:cptp_second_order} is based on a few very technical results, here we want to present an pedagogical derivation of the same result in the special case of bipartite quantum systems. Note that this derivation is quite relevant in practice as it focuses on the most studied case of a bipartite system, but the cases included in the assumptions of Theorem \ref{thm:cptp_second_order} are more general. This exact case will be the focus of the two applications of interest studied in Sec \ref{sec:spin_boson} and \ref{sec:central_spin}.

Let us assume that $\Hc = \Hc_S\otimes\Hc_E$ and let assume that the CPTP projector $\Pc$ is given by the $\Pc(\cdot) = \tr_E(\cdot)\otimes\tau$ for some full rank state $\tau\in\Df(\Hc_E)$. The projector $\Pc$ is then factorized into the reduction map $\Rc(\cdot) = \tr_E(\cdot)$ and the injection map $\Jc(\cdot) = \cdot\otimes \tau$. 
In line with the assumptions of this section we assume that $\Lc(\cdot) = -i \ad_H(\cdot)$. Given the bipartite structure of the Hilbert space $\Hc$ we can always write $H = \sum_k S_k \otimes E_k$ for two sets of Hermitian operators $\{S_k\}\subset\Bf(\Hc_S)$ and $\{E_k\}\subset\Bf(\Hc_E)$.

We can then proceed to compute the reduced generator approximated at second order $\Fc_{t,2}(\mu)$. First, using \cite[Proposition 5]{grigoletto2024exact}, we have that the generator at first order is \[\Fc_{(1)}(\mu) = \Rc \Lc \Jc(\mu) = -i[\check{H},\mu]\] where
\begin{align*}
    \check{H} &= \Jc^\dag(H) = \tr_{E}[(\one_S\otimes\tau) H]
    = \sum_k a_k(\tau) S_k
\end{align*}
and $a_k(\tau) \equiv \tr[\tau E_k]. $ Next, we can proceed to compute $\Fc_{(2)}$ by first computing $\Rc \Lc^2 \Jc$ obtaining 
\begin{align*}
    \Rc \Lc^2 \Jc(\mu) &= &= \sum_{j,k} c_{j,k}(\tau) \left[2S_j \mu S_k - \{ S_kS_j, \mu\}  \right]
\end{align*}
where we defined $c_{j,k}(\tau) \equiv \tr[ E_j\tau E_k]$.
On the other hand, since $(\Rc\Lc\Jc)^2 = \Dc_{\sqrt{2}\check{H}}$, we have
\begin{align*}
    (\Rc\Lc\Jc)^2(\mu) = \sum_{j,k}  a_j(\tau)a_k(\tau) \left[ 2 S_j \mu S_k - \{S_kS_j,\mu\} \right]
\end{align*}
and thus 
\begin{align*}
    \Fc_{(2)}(\mu)&=[\Rc \Lc^2 \Jc  - (\Rc\Lc\Jc)^2](\mu) =\\ &=\sum_{j,k} [\chi(\tau)]_{j,k} \left[2S_j \mu S_k -  \{ S_kS_j, \mu\}  \right]
\end{align*}
where we defined the matrix $\chi(\tau)$ whose elements are $[\chi(\tau)]_{j,k} \equiv c_{j,k}(\tau)-a_j(\tau)a_k(\tau)$.

One can then notice that the matrix $\chi(\tau)$ coincides with a covariance matrix, i.e. 
\begin{align*}
    [\chi(\tau)]_{j,k} &= \tr[E_j \tau E_k] - \tr[E_j\tau]\tr[E_k\tau]\\
    &= \Eb_\tau[(E_k - \Eb_\tau[E_k]\one_E) (E_j - \Eb_\tau[E_j]\one_E)]
\end{align*}
hence $\chi(\tau)$ is a positive semidefinite matrix.
Since $\chi(\tau)$ is positive semidefinite, it is diagonalizable with a unitary matrix, i.e. $UU^\dag =U^\dag U = \one$ such that $U\chi(\tau) U^\dag = \gamma(\tau)$ where $\gamma(\tau)$ is a diagonal matrix of positive values (note that in reality $U$ is also a function of $\tau$). We can then define 
\[L_h = \sum_j \overline{[U]_{h,j}} S_j \]
to obtain 
\begin{align*}
\Fc_{(2)}(\mu) 
&=\sum_{j,k} [\chi(\tau)]_{j,k} \left[2S_j \mu S_k - \{ S_kS_j, \mu\}  \right]\\
&=\sum_{j,k,h} \underbrace{[U^\dag]_{j,h}[\gamma(\tau)]_{h} [U]_{h,k}}_{[\chi(\tau)]_{j,k}} \left[2 S_j \mu S_k - \{ S_k S_j, \mu\}  \right]\\
&=\sum_{j,k,h} [\gamma(\tau)]_{h} \overline{[U]_{h,j}}[U]_{h,k} \left[2 S_j \mu S_k - \{ S_k S_j, \mu\}  \right]\\
&=\sum_h [\gamma(\tau)]_{h} \left[2L_h \mu L_h^\dag - \{ L_h^\dag L_h, \mu\}  \right] \\
& = \sum_h [\gamma(\tau)]_h \Dc_{\sqrt{2} L_h}(\mu).
\end{align*}
Because $\chi(\tau)$ is positive semidefinite and thus $[\gamma(\tau)]_h$ are non-negative, then $\Fc_{(2)}$ is of Lindblad type. This implies that the approximated time-dependent generator at second order 
\[\Fc_{t,2}(\mu) = -i[\check{H},\mu] + \sum_h t [\gamma(\tau)]_h \Dc_{\sqrt{2} L_h}(\mu)\]
is also of Lindblad type at all times $t\geq0$ as proved in Theorem \ref{thm:cptp_second_order}.

Note that in deriving the reduced generator $\Fc_{t,2}$ there was no need to move the system description into a rotating frame thus proving one of the benefits of the proposed approach. 

\section{Applications to quantum models}
\label{sec:applications}
In this section, we consider a few examples of interest:\\ 
I) The first example considers a qubit coupled to a bosonic bath. As the exact solution of such a model is well known, this example allows us to assess the goodness of the approximation proposed in this work. \\
II) Second, a qubit coupled to a finite bath is studied. Given the fact that the bath is finite-dimensional, this is typically considered a ``highly non-Markovian model''. This example allows us to showcase the numerical procedure at higher orders for both a case where the original dynamics is purely Hamiltonian and a case where a dissipative term is included. \\
III) Finally, an Ising spin chain is studied. This example allows us to show that the proposed procedure works in cases that are more general than the bipartite case.

\begin{figure*}
    \centering
    \begin{subfigure}{0.48\textwidth}
        \centering
        \includegraphics[width=\linewidth]{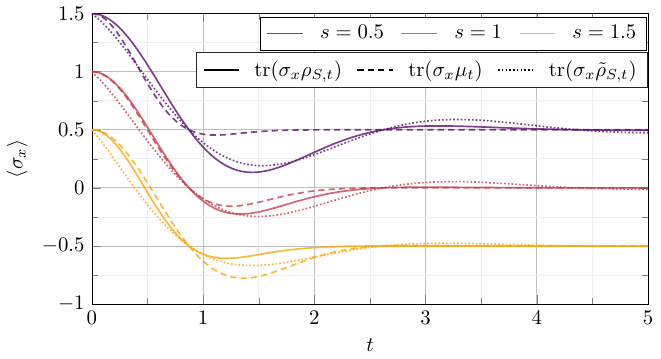}
    \end{subfigure}
    \begin{subfigure}{0.48\textwidth}
        \centering
        \includegraphics[width=\linewidth]{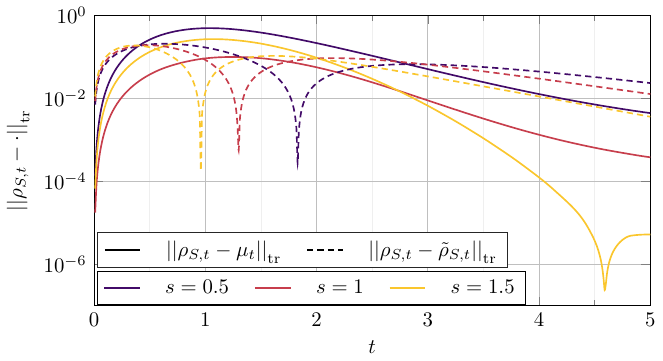}
    \end{subfigure}
    \caption{Simulation of the dephasing spin-boson model for $s=0.5,1,1.5$ where we compare the exact solution $\rho_{S,t}$ given in Eq.\eqref{eq:exact_boson}, the second order approximation $\mu_t$ derived in this work, given in Eq. \eqref{eq:time_ordered_boson} and the coarse-grained dynamics $\tilde{\rho}_{S,t}$ given in Eq. \eqref{eq:coarse_grained_boson}. Left: Evolution of the expectation value $\expect{\sigma_x}$ for the exact solution (continuous curves), second order approximation (dashed curve) and coarse-grained dynamics (dotted curve). The expectations values of the cases $s=0.5,1.5$ have been shifted by $\pm0.5$ respectively for graphical purposes. Right: evolution of the trace-norm error with respect to the exact solution committed by the second order approximation (continuous curve) and coarse-grained dynamics (dashed curve).   }
    \label{fig:dephasing_spin_boson}
\end{figure*}

\subsection{Dephasing spin boson model}
\label{sec:spin_boson}

Inspired by \cite{lidar2001completely, reina2002decoherence, breuer2002theory, lidar2020lecturenotestheoryopen}, in this section we consider a dephasing spin-boson model to further showcase the advances and drawbacks of the proposed method. Such a model is typically used to model the dephasing of qubits in quantum computers \cite{sun2025quantum}. Note that: i) in this section we make an exception to the assumption of finite-dimensional systems considered throughout the paper; ii) the model we are going to consider is exactly solvable hence the intent of this subsection is to showcase that the proposed approximation method provides results comparable to other approximation techniques, not to compare it against the exact solution of the model. 
Nevertheless, this example helps us to understand the degrees of approximation that the proposed method offers and how the process unfolds.

The model is composed of a spin $\um$ system $\Hc_S\simeq\Cb^2$ coupled to a bath $\Hc_B$ composed of $N_B$ bosonic modes, $\Hc = \Hc_S\otimes \Hc_B$. The model's dynamics is driven by the total Hamiltonian $H= H_S\otimes \one_B + \one_S \otimes H_B + H_{\rm int}$ where 
\begin{align}
    H_S &= \frac{g}{2}\sigma_z\\
    H_B &= \sum_{k=1}^{N_B} \omega_k \left(n_k+\frac{\one}{2}\right)\\
    H_{\rm int} &= \sigma_z \otimes \left(\sum_{k=1}^{N_B}\lambda_k b_k + \overline{\lambda}_k b_k^\dag\right)
\end{align} 
and where $n_k = b_k^\dag b_k$ and $b_k$ are the bosonic number and annihilation operator for mode $k$ respectively and satisfy $[b_k,b_l^\dag]=\delta_{k,l}\one$, $\omega_k\in\Rb$ is the frequency of the mode and $\lambda_k\in\Cb$ is the strength of the coupling between the mode and the spin.  

We assume that the initial state is factorized $\rho_0 = \rho_{S,0}\otimes\rho_B$ with $\rho_B = e^{-\beta H_B}/\tr[e^{-\beta H_B}]$ and some state $\rho_{S,0}\in\Df(\Hc_S)$ for the spin. We further assume that the projector $\Pc$ is $\Pc(\rho) = \tr_B(\rho)\otimes \rho_B$, which can be factorized into $\Rc(\rho) = \tr_B(\rho)$ and $\Jc(\sigma)=\sigma\otimes\rho_B$. Then we can verify $\Pc(\rho_0) = \rho_0$. 

This model has a known exact analytical solutions given by \cite[eq. (336) and (374)]{lidar2020lecturenotestheoryopen}:
\[\rho_{S,t}' = \begin{bmatrix}
[\rho_{S,0}]_{0,0}&e^{-2\gamma(t)t}[\rho_{S,0}]_{0,1}\\
e^{-2{\gamma(t)}t}[\rho_{S,0}]_{1,0}&[\rho_{S,0}]_{1,1}\\
\end{bmatrix}\]
where $\rho_{S,t}'$ is the state of the central spin in the rotating frame and  
\(\gamma(t) = \sum_k \alpha_k \frac{1-\cos(\omega_k t)}{t}\)
with $\alpha_k \equiv 2\frac{|\lambda_k|^2}{\omega_k^2}\coth\left(\frac{\beta\omega_k}{2}\right)\in\Rb.$
By taking the derivative of an off diagonal element we obtain:
\[\frac{d}{dt}[{\rho}_{S,t}']_{0,1} = -2\underbrace{\left[\gamma(t) + \frac{d\gamma(t)}{dt}t\right]}_{\equiv \xi(t)}[{\rho}_{S,t}']_{0,1}\]
where \(\xi(t) = \sum_k\alpha_k\omega_k\sin(\omega_k t)\).
This dynamics corresponds to the differential equation 
\( \frac{d}{dt}\rho_{S,t}' = \xi(t)\Dc_{\sigma_z}[\rho_{S,t}']\)
for the state in rotating frame or, equivalently, to
\begin{equation}
    \dot{\rho}_{S,t} = -i[H_S,\rho_{S,t}] + \xi(t) \Dc_{\sigma_z}(\rho_{S,t})
    \label{eq:exact_boson}
\end{equation} 
for the state in the lab frame. 

By following the procedure detailed in Sec. \ref{sec:second_order_bipartite} we can then proceed to compute the reduced time-dependent generator at second order $\Fc_{t,2}$. 
Denoting $\expect{X}_B \equiv \tr[\rho_B X]$, we recall that \cite[eq. (331)]{lidar2020lecturenotestheoryopen}: 
\begin{align*}
    \expect{b_k^\dag}_B &= \expect{b_k}_B = \expect{b_k^\dag b_l^\dag}_B = \expect{b_k b_l}_B = 0,\\
    \expect{b_j^\dag b_k}_B &= \frac{\delta_{j,k}}{e^{\beta\omega_k}-1}.
\end{align*}
We first compute the reduced Hamiltonian $\check{H} = \Jc^\dag(H)$ where $\Jc^\dag(X) = \tr_B[X(\one_S\otimes\rho_B)]$. Noting that $\tr_B[H_S\otimes\one_B(\one_S\otimes\rho_B)] = H_S \tr[\rho_B] = H_S$, $\tr_B[\one_S\otimes H_B(\one_S\otimes\rho_B)] = \one_S \expect{H_B}_B$ and
\begin{align*}
    &\tr_B\left[\sigma_z\otimes\left(\sum_{k=1}^{N_B}\lambda_k b_k + \overline{\lambda}_k b_k^\dag\right)(\one_S\otimes\rho_B)\right]\\ &\qquad = \sigma_z \left( \sum_{k=1}^{N_B}\lambda_k \cancel{\expect{b_k}_B} + \overline{\lambda}_k \cancel{\expect{b_k^\dag}_B} \right) = 0
\end{align*}
we have \(\check{H} = H_S \) as the term \(\one \expect{H_B}_B\) only contributes to a phase and can thus be removed, leading to the first-order term $\Fc_{(1)}(\mu) = -i[\check{H},\mu]$. The second order term is instead found to be $\Fc_{(2)}(\mu) = \varphi \Dc_{\sigma_z}(\mu)$ with $\varphi\equiv \sum_{k=1}^{N_B} \frac{2|\lambda_k|^2}{e^{\beta \omega_k}-1}$. All necessary calculations are detailed in Appendix \ref{sec:spin_boson_model_calculations} for the sake of clarity. 

Combining both terms we obtain an approximate reduced model described by the time-dependent Lindblad equation
\begin{equation}
    \dot{\mu}_t = -i[H_S,\mu_t] + \varphi t \Dc_{\sigma_z}(\mu_t).
    \label{eq:time_ordered_boson}
\end{equation} 
By comparing this equation with the exact solution of Eq. \eqref{eq:exact_boson} we can observe that both the Hamiltonian and noise operator are correctly recovered whereas the time-dependent dissipation strength $\xi(t)$ is approximated with $\varphi t$. Interestingly, the value $\varphi$ does not coincide with the first order Taylor expansion of $\xi(t)$ which is $\sum_k\alpha_k\omega_k$.

To better asses the quality of the reduced second order approximation we next compare the evolution of $\rho_{S,t}$ and that of $\mu_t$ with the coarse-grained approximation discussed in \cite[eq. (332)]{lidar2020lecturenotestheoryopen}, i.e. \(\tilde{\rho}_{S,t}\in\Df(\Hc_S)\) and
\begin{equation}
    \dot{\tilde{\rho}}_{S,t} = -i[H_S,\tilde{\rho}_{S,t}] + \gamma(\tau)\Dc_{\sigma_s}(\tilde{\rho}_{S,t})
    \label{eq:coarse_grained_boson}
\end{equation}
where the constant $\tau\in\Rb$ is specified later. 

The bath coupling coefficients are going to be defined as $\lambda_k = \sqrt{J(\omega_k)}$ \cite{tamascelli2020quantum} where $J(\omega)$ is the bath spectral density and is chosen to be 
\begin{equation}
    J(\omega) = \Lambda \omega_c^{1-s} \omega^s e^{-\frac{\omega}{\omega_c}}
\end{equation}
where $\Lambda\in\Rb_+$ is an overall constant, $\omega_c\in\Rb_+$ is the bath cut-off frequency and $s\in\Rb_+$ is the Ohmicity parameter which determines three cases of interest, i.e. Ohmic ($s=1$), sub-Ohmic($s<1$) and super-Ohmic ($s>1$) \cite{PhysRevB.92.195143}.

\begin{figure*}
    \centering
    \begin{subfigure}{0.48\textwidth}
        \centering
        \includegraphics[width=.8\linewidth]{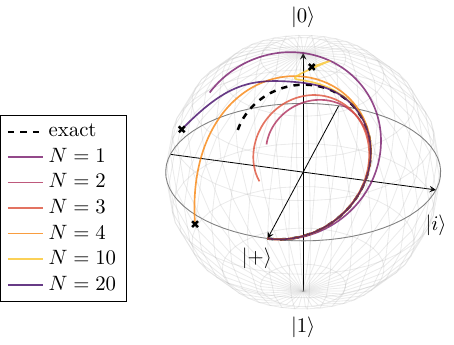}
    \end{subfigure}
    \begin{subfigure}{0.48\textwidth}
        \centering
        \includegraphics[width=.8\linewidth]{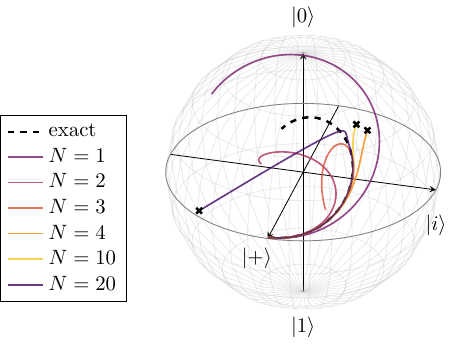}
    \end{subfigure}
    \begin{subfigure}{0.48\textwidth}
        \includegraphics[width=\linewidth]{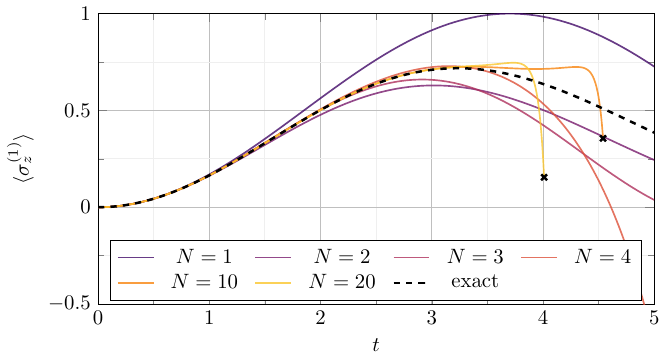}
    \end{subfigure}
    \begin{subfigure}{0.48\textwidth}
        \includegraphics[width=\linewidth]{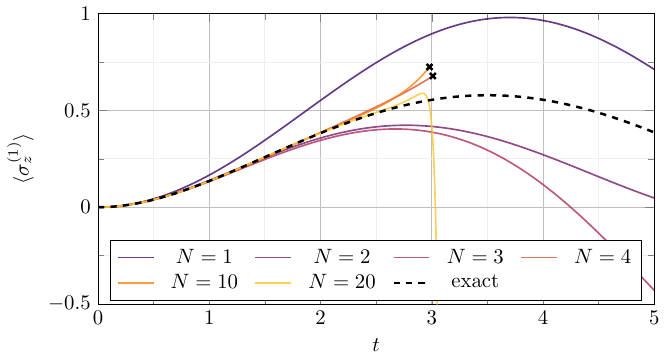}
    \end{subfigure}
    \begin{subfigure}{0.48\textwidth}
        \includegraphics[width=\linewidth]{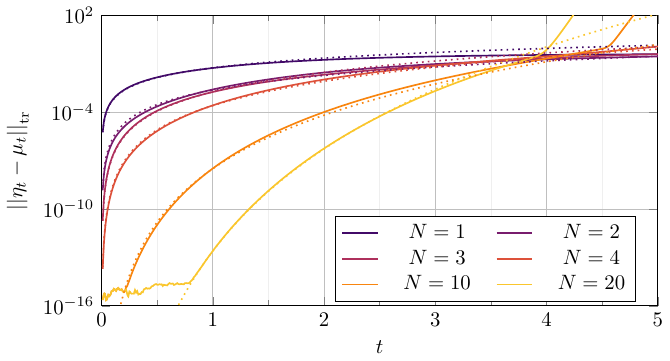}
    \end{subfigure}
    \begin{subfigure}{0.48\textwidth}
        \includegraphics[width=\linewidth]{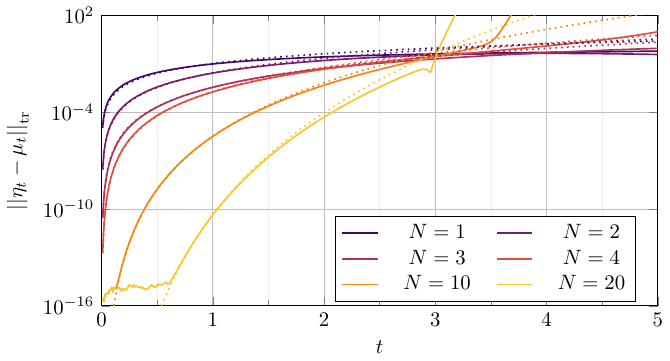}
    \end{subfigure}
    \caption{Simulations of the dissipative central spin model in the purely Hamiltonian setting (\(\Lambda=0\)) on the left and in the dissipative setting (\(\Lambda=0.8\)) on the right. The exact evolution and various degrees of approximations (\(N=1,2,3,4,10,20\)) are shown in each figure. The first row shows the trajectory of the state of the central state \(\rho_S\) in the Bloch sphere; the second row shows the magnetization of the central spin in the $\sigma_z$ direction versus time; the third row shows the trace norm error between the exact trajectory $\eta_t$ and the reduced one $\mu_t$, i.e. $\norm{\eta_t-\mu_t}_{\tr}$ versus time for different values of the approximation $N$ and with dotted curve representing a fitted $\alpha t^{N+1}$ curve. The black crosses denote points where the evolved state $\mu_t$ exits the set of density operators.}
    \label{fig:central_spin}
\end{figure*}

In Fig. \ref{fig:dephasing_spin_boson} the results of the simulations are shown in the three cases, i.e. $s=0.5,1,1.5$. The other parameters used in the simulations are $g=1.8$, $\omega_c =0.2$, $\Lambda=0.5$, $\beta=10$, $\tau=32/\omega_c$, and where $\omega_k$ are $N_B=100$, taken equally spaced between $0$ and $1$. Instead, the initial state has been fixed to $\rho_{S,0}=\mu_0 = \tilde{\rho}_{S,0} = \ketbra{+}{+}$. 
Importantly, one can observe in Fig. \ref{fig:dephasing_spin_boson} (right) that, in each of the three cases, the error committed by the second-order approximation, Eq. \eqref{eq:time_ordered_boson} is below that of the coarse-grained dynamics, Eq. \eqref{eq:coarse_grained_boson} for both the initial and final phases, while it is above for a central phase.     
This indicates that the second-order approximation performs better than the coarse-grained approximation in this case.

\subsection{Dissipative central spin model}
\label{sec:central_spin}

In this section we consider a dissipative central spin model to showcase the capabilities of the proposed technique. Similar models have been studied in the literature, see e.g. \cite{barnes2012nonperturbative, tiwari2025strong, jing2018decoherence, ferraro2008}. The model is composed by a system, i.e. $\Hc_S\simeq\Cb^2$ also denoted as spin $1$ and a bath composed of $N_B$ other spins, i.e. $\Hc_B\simeq\Cb^{2^{N_B}}$ counted from $2$ to $N_B+1$. 

\begin{figure*}[th]
    \centering
    \begin{subfigure}{0.49\textwidth}
        \centering
        \includegraphics[width=\linewidth]{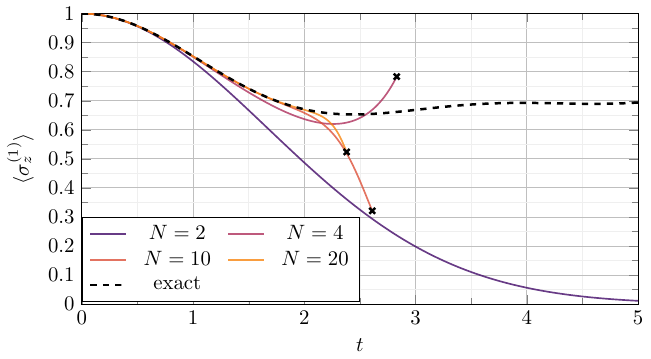}
    \end{subfigure}
    \begin{subfigure}{0.49\textwidth}
        \centering
        \includegraphics[width=\linewidth]{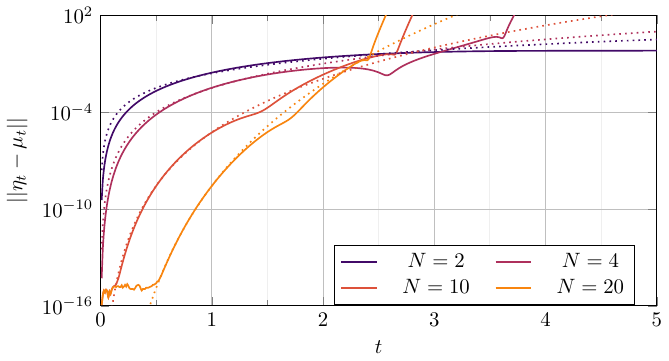}
    \end{subfigure}
    \caption{Simulation of the Ising spin chain. Magnetization of the first spin of the chain $\expect{\sigma_z^{(1)}}$ versus time on the left and the error norm $\norm{\eta_t-\mu_t}_2$ versus time on the right. Various approximation's degrees $N=2,4,10,20$ are shown in each figure. The black crosses denote the points where the evolved state $\mu_t$ exits the $n$-dimensional simplex, i.e. $\mu_t$ is no longer a probability distribution while the dotted curves on the right again represent fitted curves of $\alpha t^{N+1}$. }
    \label{fig:spin_chain}
\end{figure*}

The system Hamiltonian is described as $H = H_S + H_B + H_{\rm int}$ where:
\begin{align*}
    H_S &= \delta(\sigma_x^{(1)} + \sigma_z^{(1)})\\
    H_B &= \frac{\gamma}{4}\left(2J_x^2-\frac{N_B}{2}\one\right)\\
    H_{\rm int} &= \frac{\lambda}{2}\left(A_x \sigma_x^{(1)}J_x+ A_y \sigma_y^{(1)}J_y+A_z \sigma_z^{(1)}J_z\right)
\end{align*}
where $\sigma_u^{(j)}$ denotes the $u\in\{x,y,z\}$ Pauli matrix acting on the $j$-th spin and $J_u=\frac{1}{2}\sum_{k=2}^{N_B+1}\sigma_u^{(k)}$ denotes the total bath angular momentum operators. The bath is also in contact with a Markovian environment that dissipates information. This interaction is described by the noise operators $L_k = \Lambda\sigma_-^{(k)}$ for all $k=2,\dots,N_B+1$ where $\sigma_\pm^{(k)} = \frac{1}{2}(\sigma_x^{(k)}\pm\sigma_y^{(k)})$.

For this example we consider as an initial condition the factorized state $\rho_0 = \ketbra{+}{+} \otimes \rho_B$
where \[\rho_B = \frac{e^{-\beta H_B}}{\tr[e^{-\beta H_B}]}\]
is the thermal state at inverse temperature $\beta$. 
This allows us to define the reduction and injection maps as $\Rc(\cdot) = \tr_B(\cdot)$ and $\Jc(\cdot) = \cdot\otimes\rho_B$. Note that with this choice we have that $\Pc(\rho_0)=\rho_0$.
The reduced state then corresponds to the state of the central spin, i.e. 
\[\eta_t = \rho_S(t) \equiv \tr_B(\rho_t).\]

In Fig.\ref{fig:central_spin} the results of the simulation for this dissipative central spin model are shown. The simulation has been run in two settings: a purely Hamiltonian setting ($\Lambda=0$ on the left) and a dissipative setting ($\Lambda=0.8$ one the right). The rest of the parameters used were $N_B = 3$, $\delta=0.3$, $\lambda=0.1$, $\gamma=1$, $A_x=1.2$, $A_y=1.5$ $A_z=1.3$ and $\beta=50$. In the numerical simulations we can observe a few interesting aspects: 1) as expected, the error between the exact solution and the approximate solution $\norm{\eta_t-\mu_t}_tr$ decreases when increasing the approximation degree $N$; 2) the trajectory for approximation orders $N=1,2$ and even $3$ do not exit the Bloch sphere (note that we have proven the complete positivity for orders $N=1,2$ but not for $N=3$); 3) the simulations for $N>3$ are displaying Runge's phenomena at larger times, which is quite typical in these situations; 4) after some time, the evolved state $\eta_t$ for an approximation degree $N\geq4$ exits the Bloch sphere, indicating that the dynamics is not CPTP at all times for approximation degrees higher than $N\geq4$.

\subsection{Ising spin chain}
\label{sec:ising_model}

Although the idea of approximating a large model with a smaller one that is LTV is common in the setting of a bipartite system-environment couple, this is far from the only case where such an approach is useful. More in general, the same approaches can be adapted to cases where one wants to describe the evolution of only a small set of variables of interest of the original model. For this reason, in this section we showcase the use of the same methods presented before to reproduce the probability distribution in the standard basis. 

In order to show that the proposed method works also in cases that are not bi-partite we here consider a spin-chain model. More specifically we consider a model composed of $N_B$ spins, i.e. $\Hc\simeq\Cb^{n}$ where $n = 2^{N_B}$ whose dynamics is driven by the Ising Hamiltonian 
\begin{equation}
    H = \sum_{j=1}^{N_B} h\sigma_z^{(j)} + A \sum_{j=1}^{N_B-1} \sigma_{x}^{(j)}\sigma_{x}^{(j+1)}.
\end{equation}    
In this example, we consider $\Pc$ to be the CPTP projector on the diagonal $\Pc(\rho) = \sum_{j=0}^{n-1} \ketbra{j}{j}\rho\ketbra{j}{j}$ which can be factorized into $\Rc(\rho) = \sum_{j=0}^{n-1} \ketbra{j}{j}\rho \ket{j}$ and $\Jc(\eta)=\sum_{j=0}^{n-1} \ketbra{j}{j} \bra{j}\eta$. Note that, in such a case, $\eta=\Rc(\rho)\in\Rb^n$ is a probability distribution, i.e. $\eta_i\geq0$ and $\onev^T\eta=1$ where $\onev\in\Rb^n$ is the vector of all ones. This implies that $\Fc_t\in\Rb^{n\times n}$. 

As an initial condition we consider $\rho_0=\ketbra{0}{0}$. Trivially, $\Pc(\rho_0)=\rho_0$.

Numerically, we verified that $\Fc_{(1)}=\Fc_{(3)}=0$ while $\Fc_{(2)}$ is graphically represented in Fig. \ref{fig:matrix} where one can verify that $\Fc_{(2)}$ is Metzler (a matrix $A$ is Metzler if $A_{i,j}\geq0$ for all $i\neq j$) and is such that $\onev^T \Fc_{(2)} = 0$. Those two conditions are the classical equivalent to the properties characterizing a Lindblad generator and in fact imply that $e^{\Fc_{(2)}t}$ is a stochastic matrix for all $t\geq0$. This directly implies that $\Tb e^{\int_0^t \Fc_{s,2} ds}$ is CPTP, or stochastic in this case, for all $t\geq0$. 

\begin{figure}[th]
    \centering
    \includegraphics[width=.6\linewidth]{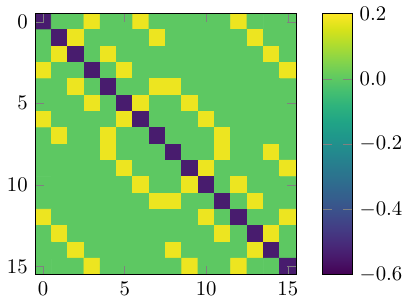}
    \caption{Graphical representation of the term $\Fc_{(2)}$ for the time-dependent generator $\Fc_t$ of the Ising spin chain model.}
    \label{fig:matrix}
\end{figure}

Fig. \ref{fig:spin_chain} shows the results of the simulations of the Ising spin chain model with $N_B = 4$, $A= 0.3$, $h=0.36$ where we can observe similar results to the other simulations included in this work: increasing the degree of approximation decreases the error $\norm{\eta_t-\mu_t}_2$ at low times; for $N=2$ the approximate evolved state $\mu_t$ remains a probability distribution, while for $N\geq4$, after a certain time, the state is no longer a probability distribution. 

\section{Conclusions and outlook}
\label{sec:conclusion}

In this work, we explored a novel method for approximating the dynamics of a large autonomous system projected onto a given subspace with a time-dependent generator. The problem is inspired by the derivation of TCL ME in quantum physics but has been treated here in the general case of linear autonomous systems. 

The main contribution of this work consists of a novel method for computing the coefficients of the time-dependent model that is polynomial in time and that approximates the evolution of the projected dynamics for small times. Importantly, the coefficients of the polynomial can be determined using a recursive formula.  

Benefits of the proposed reduction method include: the method does not require to move into a rotating frame, making the calculations simpler; the iterative form of the coefficients allows for easy calculations of the terms; for lower orders the method offers an advantage with respect to truncating the exponential, as shown in Sec. \ref{sec:linear_testbed}; the method does not rely on a perturbation on a small parameter $\varepsilon$ allowing for analysis of systems in strong coupling regimes; when the original dynamics is purely Hamiltonian, the reduced time-dependent model at the lowest two orders is ensured to generate a CPTP two-parameter semigroup at all times. All of this comes at the cost of having a model that approximates the original evolution in the short-time regime. 

Future work shall focus on: 
i) extending Theorem \ref{thm:time_ordered_expansion} and Theorem \ref{thm:approx} to the case of analytic generators; 
ii) lifting the assumption $Px_0 =x_0$; 
iii) generalize the results presented here to include control inputs (both linear and bilinear);
iv) apply the proposed procedure to other quantum models in the strong coupling regime to study models that were not accessible before and further test the capabilities of the method; 
v) combine the classical model reduction methods and the time-dependent model reduction proposed here to see if the combination of the two methods can allow to improve the accuracy of the reduced model. Note that a first step in this direction has been taken in the physics literature by \cite{breuer2006,breuer2007}, where, including extra degrees of freedom, has been shown to improve the accuracy of the derived TCL ME.

Given the numerical examples of the central spin model presented in Sec. \ref{sec:central_spin} we do not expect the reduced model to retain the CP property for orders higher than $N=3$. However, it is possible that, in the Hamiltonian case, for $N=3$ one can prove that the evolution of the reduced model is necessarily CP at all times, as none of the simulation carried out while preparing this work contradicted this statement.

\section{Acknowledgments}
The author would like to thank Francesco Ticozzi for his contribution in the initial idea of this work and for his constant support. The author also thanks Lorenza Viola, Michiel Burgelman, Augusto Ferrante, and Giacomo Baggio for useful discussions on the matters of this work.

\paragraph{Code availability:} The code necessary to reproduce all the numerical experiments presented in this work is available at \cite{Myrepo}. 

\paragraph{Funding:}
The author has been supported by the Italian Ministry of University and Research under the PRIN project “Extracting essential information and dynamics from complex networks”, grant no. 2022MBC2EZ.

\bibliography{ref}

\onecolumngrid
\appendix
\input{appendix.tex}

\end{document}

%% file: header.tex
\usepackage[usenames,dvipsnames]{xcolor}

\usepackage{cancel}
\usepackage{amsfonts,amsmath,amssymb}
\usepackage{calc}
\usepackage[english]{babel}
\usepackage[normalem]{ulem}
\usepackage[hidelinks]{hyperref}

\usepackage{amsthm}
\usepackage{graphicx}

\usepackage{quiver}

\usepackage[linesnumbered,noend]{algorithm2e}

\usepackage{subcaption,caption}
\usepackage{epsfig}
\usepackage{ulem}
\usepackage{tikz}
\usepackage{mathrsfs}
\usepackage{bbold}             

\usepackage{soul}

\usepackage{pgfplots}

\pgfplotsset{compat = newest}

\usepackage{bm}
\usetikzlibrary{shapes,arrows}
\usetikzlibrary{matrix}
\usetikzlibrary{calc,patterns,angles,quotes,arrows,automata}
\graphicspath{ {./imgs/}}

\makeatother
\usepackage{thmtools}

\theoremstyle{definition}
\newtheorem{example}{Example}

\theoremstyle{theorem}

\newtheorem{proposition}{Proposition}

\newtheorem{theorem}{Theorem}

\theoremstyle{definition}

\theoremstyle{remark}

%% file: commands.tex
\newcommand{\Dc}{\mathcal{D}}

\newcommand{\Fc}{\mathcal{F}}

\newcommand{\Hc}{\mathcal{H}}
\newcommand{\Ic}{\mathcal{I}}
\newcommand{\Jc}{\mathcal{J}}
\newcommand{\Kc}{\mathcal{K}}
\newcommand{\Lc}{\mathcal{L}}

\newcommand{\Pc}{\mathcal{P}}

\newcommand{\Rc}{\mathcal{R}}

\newcommand{\Uc}{\mathcal{U}}

\newcommand{\Xc}{\mathcal{X}}

\newcommand{\As}{\mathscr{A}}

\newcommand{\Ns}{\mathscr{N}}

\newcommand{\Vs}{\mathscr{V}}

\newcommand{\Cb}{\mathbb{C}}

\newcommand{\Eb}{\mathbb{E}}

\newcommand{\Rb}{\mathbb{R}}

\newcommand{\Tb}{\mathbb{T}}

\newcommand{\Bf}{\mathfrak{B}}

\newcommand{\Df}{\mathfrak{D}}

\newcommand{\es}{\mathcal{J}}

\newcommand{\um}{\frac{1}{2}}
\newcommand{\one}{\mathbb{1}}

\newcommand{\onev}{{\bm 1}}

\newcommand{\ad}{{\rm ad}}

\newcommand{\tr}{{\rm tr}}

\newcommand{\Span}{{\rm span}}

\newcommand{\ket}[1]{\left| #1 \right>}
\newcommand{\bra}[1]{\left< #1 \right|}
\newcommand{\norm}[1]{\left|\left| #1 \right|\right|}
\newcommand{\inner}[2]{\left< #1,#2 \right>}

\newcommand{\expect}[1]{\left< #1 \right>}

\newcommand{\ketbra}[2]{\left\vert #1 \middle>\middle< #2 \right\vert}

\definecolor{darkviolet}{rgb}{0.58, 0.0, 0.83}
\definecolor{lavender}{rgb}{0.45, 0.31, 0.59}

%% file: appendix.tex
\section{Technical results and proofs for the time-ordered expansion}
\label{sec:thechnical_results:time_ordered}
Before proving the main result, Theorem \ref{thm:time_ordered_expansion}, we shall develop some intuition about how we prove it.
\begin{example}
To provide some intuition about Theorem \ref{thm:time_ordered_expansion} let us consider a simple case: \(F_t = F_{(1)} + t F_{(2)} +t^2 F_{(3)}\). 
Let us first compute the first few terms of the Dyson expansion \cite{Dyson} of $\Tb e^{\int_0^t F_s ds}$, i.e. $\Tb e^{\int_0^t F_s ds} = \sum_{n=0}^{+\infty} \Phi_{(n),t}$.
At zeroth order we have $\Phi_{(0),t} = \Ic$ while at orders $1,2$ and $3$ we have:
\begin{align*}
    \Phi_{(1),t} &= \int_0^t F_{s} ds = \int_0^t (F_{(1)} + sF_{(2)} + s^{2} F_{(3)}) ds = 
    t F_{(1)} + t^2\frac{F_{(2)}}{2} + t^3 \frac{F_{(3)}}{3} ,
\end{align*}
\begin{align*}
    \Phi_{(2),t} &= \int_0^t dt_1\, \int_0^{t_1} dt_2\, F_{t_1} F_{t_2} 
    = \sum_{j=0}^{1} \int_0^t dt_1\, F_{t_1} F_{(j+1)} \int_0^{t_1} dt_2\,  t_2^j 
    = \sum_{j=0}^{1} \int_0^t dt_1\, F_{t_1}F_{(j+1)}  \frac{t_1^{j+1}}{j+1}\\
    & = \sum_{k,j=0}^{1} F_{(k+1)}F_{(j+1)} \int_0^t dt_1\, t_1^k \frac{t_1^{j+1}}{j+1} = \sum_{k,j=0}^{1} F_{(k+1)}F_{(j+1)} \frac{t^{k+j+2}}{(k+j+2)(j+1)} = \sum_{k,j=1}^{2} F_{(k)}F_{(j)} \frac{t^{k+j}}{(j+k)j}\\
    &= \underbrace{F_{(1)}^2 \frac{t^2}{2}}_{j=k=1}
    + \underbrace{F_{(2)}F_{(1)}\frac{t^3}{3}}_{j=1,k=2} 
    + \underbrace{F_{(1)}F_{(2)}\frac{t^3}{6}}_{j=2,k=1} 
    + O(t^4),
\end{align*}
\begin{align*}
    \Phi_{(3),t} &= \int_0^t dt_1 \int_0^{t_1} dt_2 \int_0^{t_2} dt_3 F_{t_1} F_{t_2} F_{t_3} = \int_0^t dt_1 F_{t_1} \underbrace{\int_0^{t_1} dt_2 \int_0^{t_2} dt_3  F_{t_2} F_{t_3}}_{= \sum_{j,k=1}^{2} F_{(k)}F_{(j)} \frac{t_1^{k+j}}{(j+k)j}}\\ 
    &= \sum_{h,k,j=1}^{2} F_{(h)}F_{(k)}F_{(j)} \int_0^t dt_1\, t_1^h \frac{t_1^{k+j}}{(j+k)j} = \sum_{h,k,j=1}^{2} F_{(h)}F_{(k)}F_{(j)}  \frac{t^{h+k+j}}{(h+k+j)(j+k)j} 
     = \underbrace{F_{(1)}^3 \frac{t^{3}}{3!}}_{h=k=j=1} + O(t^4).
\end{align*}
Grouping together the terms of $\Tb e^{\int_0^t F_s ds}$ by the order of the exponent of $t$, we have 
\begin{align*}
    \Tb e^{\int_0^t F_s ds} &= \underbrace{\Ic}_{E_{(0)}} + t \underbrace{F_{(1)}}_{E_{(1)}} + t^2 \underbrace{\left( \frac{F_{(2)}}{2} + \frac{F_{(1)}^2}{2} \right)}_{E_{(2)}} + t^3 \underbrace{\left(\frac{F_{(3)}}{3}+\frac{F_{(1)}F_{(2)}}{6} +\frac{F_{(2)}F_{(1)}}{3} + \frac{F_{(1)}^3}{3!} \right)}_{E_{(3)}} +O(t^4)
\end{align*}
which coincides with the first few terms obtained using the equation given in Theorem \ref{thm:time_ordered_expansion}:
\begin{align*}
    E_{(0)} &= \Ic,\\
    E_{(1)} &= F_{(1)}E_{(0)} = F_{(1)},\\
    E_{(2)} &= \frac{1}{2}(F_{(1)}E_{(1)} +  F_{(2)}E_{(0)}) = \frac{1}{2}(F_{(1)}^2 + F_{(2)}),\\
    E_{(2)} &= \frac{1}{3}(F_{(1)}E_{(2)} +  F_{(2)}E_{(1)} + F_{(3)}E_{(0)}) = \frac{1}{3}(\frac{1}{2}F_{(1)}^3 + \frac{1}{2}F_{(1)}F_{(2)} + F_{(2)}F_{(1)} +F_{(3)}).
\end{align*}
\hfill\qed
\end{example}
\begin{proof}[Theorem \ref{thm:time_ordered_expansion}]
    We shall start by proving that $\Tb e^{\int_0^t F_s ds} = \sum_{n=0}^{+\infty} \Phi_{(n),t}$ where $\Phi_{(0,t)} = \Ic$ and, for $n>0$, \[\Phi_{(n),t} = \sum_{k_1,\dots,k_n=1}^{N+1} t^{\sum_{j=1}^n k_j} \frac{ F_{(k_1)} \dots F_{(k_n)}}{\prod_{j=1}^n \sum_{h=n-j+1}^n k_h}.\]
    First of all, using the Dyson series expansion for the time-ordered exponential $\Tb e^{\int_0^t F_s ds}$, see e.g. \cite{Dyson, argeri2014magnus, sakurai2020modern}, we have that: 
    \[\Tb e^{\int_0^t F_s ds} = \sum_{n=0}^{+\infty} \Phi_{(n),t} \qquad\text{ with }\qquad 
    \Phi_{(n),t} = \int_0^t dt_1\, \int_0^{t_1} dt_2\, \dots \int_0^{t_{n-1}} dt_n\, F_{t_1}F_{t_2}\dots F_{t_n}\,    .\]
    Next, by noticing that \(\Phi_{(n),t} = \int_0^t ds\, F_{s} \Phi_{(n-1),s} \) we can prove the statement by induction.
    The case $n=1$ is proven by simple calculations:
    \begin{align*}
        \Phi_{(1),t} &= \int_0^t ds\, F_{s}  = \int_0^t ds\, \sum_{k_1=0}^N s^{k_1}F_{(k_1+1)} = \sum_{k_1=0}^{N} F_{(k_1+1)} \int_0^t ds\, s^{k_1}  
        = \sum_{k_1=0}^{N} F_{(k_1+1)} \frac{t^{k_1+1}}{k_1+1} 
         = \sum_{k_1=1}^{N+1} F_{(k_1)} \frac{t^{k_1}}{k_1}.
    \end{align*}
    Note that $\int_0^t ds \sum_{k=0}^N = \sum_{k=0}^{N} \int_0^t ds$ under the assumption that $N$ is finite, i.e. $N<\infty$. 
    We next assume that the statement holds for $n$, i.e. 
    \[\Phi_{(n),t} = \sum_{k_1,\dots,k_n=1}^{N+1} t^{\sum_{j=1}^n k_j} \frac{F_{(k_1)}\dots F_{(k_n)}}{\prod_{j=1}^n \sum_{h=n-j+1}^n k_h}.\]
    Then 
    \begin{align*}
        \Phi_{(n+1),t} &= \int_0^t ds\, F_s \Phi_{(n),s} 
        = \int_0^t ds\, F_s \sum_{k_1,\dots,k_n=1}^{N+1} s^{\sum_{j=1}^n k_j} \frac{ F_{(k_n)}\dots F_{(k_1)}}{\prod_{j=1}^n \sum_{h=n-j+1}^n k_h} \\
        &= \int_0^t ds\, \sum_{k=0}^{N} s^k F_{(k+1)} \sum_{k_1,\dots,k_n=1}^{N+1}  s^{\sum_{j=1}^n k_j} \frac{ F_{(k_1)} \dots F_{(k_n)}}{\prod_{j=1}^n \sum_{h=n-j+1}^n k_h} \\
        &= \sum_{k=0}^{N} \sum_{k_1,\dots,k_n=1}^{N+1} F_{(k+1)} \frac{ F_{(k_1)}\dots F_{(k_n)}}{\prod_{j=1}^n \sum_{h=n-j+1}^n k_h}  \int_0^t s^{k+\sum_{j=1}^n k_j} ds \\
        &= \sum_{k=0}^{N} \sum_{k_1,\dots,k_n=1}^{N+1} \frac{t^{k+1+\sum_{j=1}^n k_j}}{k+1+\sum_{j=1}^n k_j} F_{(k+1)} \frac{F_{(k_1)}\dots F_{(k_n)}}{\prod_{j=1}^n \sum_{h=n-j+1}^n k_h}
    \end{align*}
    \begin{align*} \Phi_{(n+1),t}       
        &= \sum_{k=1}^{N+1} \sum_{k_1,\dots,k_n=1}^{N+1} \frac{t^{k+\sum_{j=1}^n k_j}}{k+\sum_{j=1}^n k_j} F_{(k)} \frac{F_{(k_1)}\dots F_{(k_n)}}{\prod_{j=1}^n \sum_{h=n-j+1}^n k_h} 
        = \sum_{k_1=1}^{N+1} \sum_{k_2,\dots,k_{n+1}=1}^{N+1} \frac{t^{k_1+\sum_{j=2}^{n+1} k_j}}{k_1+\sum_{j=2}^{n+1} k_j} \frac{F_{(k_1)}F_{(k_2)}\dots F_{(k_{n+1})}}{\prod_{j=1}^{n} \sum_{h=n-j+2}^{n+1} k_h}\\
        &= \sum_{k_1,\dots,k_{n+1}=1}^{N+1} \frac{t^{\sum_{j=1}^{n+1} k_j}}{\sum_{j=1}^{n+1} k_j} \frac{F_{(k_1)}\dots F_{(k_{n+1})}}{\prod_{j=1}^{n} \sum_{h=n-j+2}^{n+1} k_h}
        = \sum_{k_1,\dots,k_{n+1}=1}^{N+1} t^{\sum_{j=1}^{n+1} k_j} \frac{F_{(k_1)}\dots F_{(k_{n+1})}}{\prod_{j=1}^{n+1} \sum_{h=n-j+2}^{n+1} k_h}.
    \end{align*}
    thus proving the statement.

    We next prove that \(\Tb e^{\int_0^t F_s ds} = \sum_{f=0}^{+\infty} t^k E_{(k)}\) where $E_{(0)} = \Ic$ and, for $k>0$,
    \[ E_{(k)} = \sum_{n=1}^k \sum_{\substack{{k_1,\dots,k_{n} = 1}\\{\text{s.t. }\sum_{j=1}^{n}k_{j}=k}}}^{k} \frac{F_{(k_1)}\dots F_{(k_n)} }{\prod_{j=1}^n \sum_{h=n-j+1}^n k_h}. \]
    To prove this claim we can start by noticing that each term $\Phi_{(n),t}$ contains powers $t^k$ of order $k\geq n$. Hence, to construct $E_{(k)}$ of order $k$ we need to collect and sum all the coefficients that multiply $t^k$ among the terms $\Phi_{(n),t}$ for $1\leq n\leq k$. Given $\Phi_{(n),t}$ with $n\leq k$ we have that the coefficient multiplying $t^k$ must satisfy the constraint $\sum_{j=1}^n k_j = k$. 
    Thus, the coefficient of $\Phi_{(n),t}$ that multiply $t^k$ are
    \[\sum_{\substack{{k_1,\dots,k_{n} = 1}\\{\text{s.t. }\sum_{j=1}^{n}k_{j}=k}}}^{N+1} \frac{ F_{(k_1)}\dots F_{(k_n)}}{\prod_{j=1}^n \sum_{h=n-j+1}^n k_h} = \sum_{\substack{{k_1,\dots,k_{n} = 1}\\{\text{s.t. }\sum_{j=1}^{n}k_{j}=k}}}^{k} \frac{ F_{(k_1)}\dots F_{(k_n)}}{\prod_{j=1}^n \sum_{h=n-j+1}^n k_h}\]
    where we used the fact that, if any $k_j>k$ then $\sum_{j=1}^n k_j>k$.  
    Noticing that $\prod_{j=1}^n \sum_{h=n-j+1}^n k_h = k(\prod_{j=1}^{n-1} \sum_{h=n-j+1}^n k_h)$ and summing over $1\leq n\leq k$ we obtain, for $k>0$, the second claim, i.e.  
    \[E_{(k)} = \sum_{n=1}^{k} \sum_{\substack{{k_1,\dots,k_{n} = 1}\\{\text{s.t. }\sum_{j=1}^{n}k_{j}=k}}}^{k} \frac{ F_{(k_1)}\dots F_{(k_n)}}{k(\prod_{j=1}^{n-1} \sum_{h=n-j+1}^n k_h)}.\]

    We are now finally ready to prove the main claim of the theorem, that is: \(E_{(k)} = \frac{1}{k}\sum_{s=1}^{k}F_{(s)}E_{(k-s)}.\)
    Recalling that $E_{(0)}=\Ic$ we can write:
    \begin{align*}
        \frac{1}{k}\sum_{s=1}^{k} F_{(s)} E_{(k-s)}  &= \sum_{s=1}^{k-1} \frac{F_{(s)}}{k} E_{(k-s)} + \frac{F_{(k)}}{k} \\
        &= \sum_{s=1}^{k-1} \sum_{n=1}^{k-s} \sum_{\substack{{k_1,\dots,k_{n} = 1}\\{\text{s.t. }\sum_{j=1}^{n}k_{j}=k-s}}}^{k} \frac{F_{(s)}}{k} \frac{ F_{(k_1)}\dots F_{(k_n)}}{(k-s)(\prod_{j=1}^{n-1} \sum_{h=n-j+1}^n k_h)} + \frac{F_{(k)}}{k} \\
        &= \sum_{s=1}^{k-1}\sum_{n=1}^{k-1} \sum_{\substack{{k_1,\dots,k_{n} = 1}\\{\text{s.t. }\sum_{j=1}^{n}k_{j}=k-s}}}^{k} \frac{F_{(s)}}{k} \frac{ F_{(k_1)}\dots F_{(k_n)}}{(k-s)(\prod_{j=1}^{n-1} \sum_{h=n-j+1}^n k_h)} + \frac{F_{(k)}}{k}  
    \end{align*}
    where we used the fact that $\sum_{n=k-s+1}^{k-1} \sum_{\substack{{k_1,\dots,k_{n} = 1}\\{\text{s.t. }\sum_{j=1}^{n}k_{j}=k-s}}}^{k} \frac{F_{(s)}}{k} \frac{ F_{(k_1)}\dots F_{(k_n)}}{(k-s)(\prod_{j=1}^{n-1} \sum_{h=n-j+1}^n k_h)} = 0$ since for $n>k-s$ we have that $\sum_{j=1}^{n}k_{j} > k-s$ for any set of $k_j$. This implies that
    \begin{align*}
        \frac{1}{k}\sum_{s=1}^{k} F_{(s)}E_{(k-s)}  
        &= \sum_{n=1}^{k} \sum_{s=1}^{k-1} \sum_{\substack{{k_1,\dots,k_{n} = 1}\\{\text{s.t. }\sum_{j=1}^{n}k_{j}=k-s}}}^{k} \frac{F_{(s)}}{k} \frac{ F_{(k_1)}\dots F_{(k_n)}}{(k-s)(\prod_{j=1}^{n-1} \sum_{h=n-j+1}^n k_h)} +\frac{F_{(k)}}{k} \\
        &= \sum_{n=1}^{k-1} \sum_{k_1=1}^{k-1} \sum_{\substack{{k_2,\dots,k_{n+1} = 1}\\{\text{s.t. }\sum_{j=2}^{n+1}k_{j}=k-k_1}}}^{k} \frac{F_{(k_1)}}{k} \frac{ F_{(k_2)}\dots F_{(k_{n+1})}}{\underbrace{(k-k_1)}_{\sum_{h=2}^{n+1}k_h}(\prod_{j=1}^{n-1} \sum_{h=n-j+1}^n k_{h+1})} +\frac{F_{(k)}}{k} \\
        &= \sum_{n=1}^{k-1} \sum_{k_1=1}^{k-1} \sum_{\substack{{k_2,\dots,k_{n+1} = 1}\\{\text{s.t. }\sum_{j=1}^{n+1}k_{j}=k}}}^{k} \frac{F_{(k_1)} F_{(k_2)}\dots F_{(k_{n+1})}}{k(\prod_{j=1}^{n} \sum_{h=n-j+2}^{n+1} k_{h})} +\frac{F_{(k)}}{k} \\
        &= \sum_{n=2}^{k} \sum_{k_1=1}^{k-1} \sum_{\substack{{k_2,\dots,k_{n} = 1}\\{\text{s.t. }\sum_{j=1}^{n}k_{j}=k}}}^{k} \frac{F_{(k_1)} F_{(k_2)}\dots F_{(k_{n})}}{k(\prod_{j=1}^{n-1} \sum_{h=n-j+1}^{n} k_{h})}+\frac{F_{(k)}}{k}\\
        &= \sum_{n=1}^{k} \sum_{k_1=1}^{k} \sum_{\substack{{k_2,\dots,k_{n} = 1}\\{\text{s.t. }\sum_{j=1}^{n}k_{j}=k}}}^{k} \frac{F_{(k_1)} F_{(k_2)}\dots F_{(k_{n})}}{k(\prod_{j=1}^{n-1} \sum_{h=n-j+1}^{n} k_{h})} = E_{(k)},
    \end{align*}    
    thus concluding the proof.
\end{proof}

Note that he technical reason why we choose to assume $F_t$ to be polynomial ($N<\infty$) instead of analytic is because in this proof we used in a couple of instances the property that $\int_0^t \sum_{k=0}^N s^k F_{(k+1)} ds = \sum_{k=0}^N \int_0^t s^k F_{(k+1)} ds $ which trivially holds if $N<\infty$. To lift this assumption we would need to assume some extra properties about the operators $F_{(k+1)}$ so that the order between the summation and integral can be changed.

\section{Calculations for the spin boson model}
\label{sec:spin_boson_model_calculations}
We can write  
\[H = \underbrace{\one_S}_{\equiv S_0} \otimes \underbrace{H_B}_{\equiv E_0} + \underbrace{\sigma_z}_{S_1} \otimes \underbrace{\left(\frac{g}{2}\one + \sum_k \lambda_k b_k + \overline{\lambda}_k b_k^\dag\right)}_{\equiv E_1}\]
so that the second order term becomes $F_{(2)}(\mu) = \sum_{j,k=0}^1 \chi_{j,k}(\rho_B) \left(S_j\mu S_k - \frac{1}{2}\{S_j^\dag S_k, \mu\}\right)$
where $\chi_{j,k}(\rho_B) = \expect{E_k E_j}_B - \expect{E_j}_B \expect{E_k}_B$. 
We can then start by noting that \(\expect{E_1}_B = g/2\) and 
\begin{align*}
    \expect{H_B}_B &= \sum_{k=1}^{N_B} \omega_k \left( \expect{b_k^\dag b_k}_B + \expect{\one}_B/2\right)= \sum_{k=1}^{N_B} \omega_k \left( \frac{1}{e^{\beta \omega_k}-1} + \frac{1}{2}\right) 
    = \sum_{k=1}^{N_B} \frac{\omega_k}{2} \frac{ e^{\beta \omega_k}+1}{e^{\beta \omega_k}-1}
\end{align*} 
in order to compute
\begin{align*}
    \chi_{1,0}(\rho_B) &= \expect{\left(\frac{g}{2}\one + \sum_j\lambda_j b_j + \overline{\lambda}_jb_j^\dag\right)\sum_k\omega_k\left(n_k+\frac{\one}{2}\right)}_B - \frac{g}{2} \expect{H_B}_B\\
    &= \sum_k \frac{g}{2}\omega_k (\expect{n_k}_B + \um) +\sum_{j}\lambda_j\omega_k \cancel{(\expect{b_jn_k}_B +\expect{b_j}_B)} + \overline{\lambda}_j\omega_k\cancel{\left(\expect{b_j^\dag n_k} + \expect{b_j^\dag}_B\right)} - \frac{g}{2} \expect{H_B}_B\\
    & = \sum_k \frac{g}{2}\omega_k \left(\frac{1}{e^{\beta \omega_k}-1} + \um\right) - \frac{g}{2} \frac{\omega_k}{2} \frac{ e^{\beta \omega_k}+1}{e^{\beta \omega_k}-1} = 0
\end{align*}
as $\expect{b_jn_k}_B = \expect{b_j^\dag n_k}_B =0$.

This fact directly implies that $F_{(2)}(\mu) = \sum_{j=0}^1 \chi_{j,j}(\rho_B)\Dc_{S_j}(\mu)$ and, since $S_0=\one_S$ we have that $\Dc_{S_0} = 0$ thus recovering $F_{(2)}(\mu) = \chi_{1,1}(\rho_B)\Dc_{\sigma_z}(\mu)$. It then remains to compute the coefficient
\begin{align*}
    \chi_{1,1}(\rho_B) &= \expect{\left(\frac{g}{2}\one + \sum_k \lambda_k b_k + \overline{\lambda}_k b_k^\dag\right)\left(\frac{g}{2}\one + \sum_k \lambda_k b_k + \overline{\lambda}_k b_k^\dag\right)}_B - \frac{g^2}{4}\\
    &= \cancel{\frac{g^2}{4}\expect{\one}_B} + \sum_k g\left(\lambda_k \cancel{\expect{b_k}_B} + \overline{\lambda}_k \cancel{\expect{b_k^\dag}_B}\right) \\&\qquad +\sum_{j,k} \left(\lambda_j\lambda_k \cancel{\expect{b_jb_k}}_B + \lambda_j\overline{\lambda}_k \expect{b_jb_k^\dag}_B + \overline{\lambda}_j\lambda_k \expect{b_j^\dag b_k}_B + \overline{\lambda_j}\overline{\lambda}_k \cancel{\expect{b_j^\dag b_k^\dag}_B}\right)  - \cancel{\frac{g^2}{4}}\\
    &= \sum_{k=1}^{N_B} \frac{2|\lambda_k|^2}{e^{\beta \omega_k}-1} \equiv \varphi.
\end{align*}